\documentclass[copyright,creativecommons]{eptcs}

\providecommand{\event}{MSFP 2022}

\usepackage{amsmath,amssymb,amsthm}
\usepackage[capitalize]{cleveref}
\usepackage{mathtools}
\usepackage{stmaryrd}
\usepackage{tikz,tikz-cd}
\usepackage{underscore}

\theoremstyle{plain}
\newtheorem{theorem}{Theorem}[section]
\newtheorem{proposition}[theorem]{Proposition}

\theoremstyle{definition}
\newtheorem{definition}[theorem]{Definition}

\newtheorem{example}[theorem]{Example}

\newtheorem{remark}[theorem]{Remark}

\title{What Makes a Strong Monad?}
\author{%
  Dylan McDermott
  \institute{Reykjavik University, Iceland}
  \email{dylanm@ru.is}
  \and
  Tarmo Uustalu
  \institute{Reykjavik University, Iceland}
  \institute{Tallinn University of Technology, Estonia}
  \email{tarmo@ru.is}}
\def\titlerunning{What Makes a Strong Monad?}
\def\authorrunning{Dylan McDermott \& Tarmo Uustalu}

\DeclareFontFamily{U}{MnSymbolA}{}
\DeclareSymbolFont{MnSyA}{U}{MnSymbolA}{m}{n}
\DeclareFontShape{U}{MnSymbolA}{m}{n}{
<-6>  MnSymbolA5
<6-7>  MnSymbolA6
<7-8>  MnSymbolA7
<8-9>  MnSymbolA8
<9-10> MnSymbolA9
<10-12> MnSymbolA10
<12->   MnSymbolA12}{}
\DeclareFontShape{U}{MnSymbolA}{b}{n}{
<-6>  MnSymbolA-Bold5
<6-7>  MnSymbolA-Bold6
<7-8>  MnSymbolA-Bold7
<8-9>  MnSymbolA-Bold8
<9-10> MnSymbolA-Bold9
<10-12> MnSymbolA-Bold10
<12->   MnSymbolA-Bold12}{}
\DeclareMathSymbol{\leftpitchfork}{\mathrel}{MnSyA}{"8A}%

\newcommand\imCMsym[4][\mathord]{%
  \DeclareFontFamily{U} {#2}{}
  \DeclareFontShape{U}{#2}{m}{n}{
    <-6> #25
    <6-7> #26
    <7-8> #27
    <8-9> #28
    <9-10> #29
    <10-12> #210
    <12-> #212}{}
  \DeclareSymbolFont{CM#2} {U} {#2}{m}{n}
  \DeclareMathSymbol{#4}{#1}{CM#2}{#3}
}
\imCMsym{cmmi}{124}{\CMjmath}
\imCMsym[\mathop]{cmex}{96}{\CMcoprod}
\renewcommand{\jmath}{\CMjmath}
\renewcommand{\coprod}{\CMcoprod}

\newcommand{\catname}[1]{\mathbf{#1}}
\newcommand{\V}{\catname V}
\newcommand{\C}{\catname C}
\newcommand{\D}{\catname D}
\newcommand{\E}{\catname E}
\newcommand{\Set}{\catname{Set}}
\newcommand{\PSet}{\catname{Set}_\star}
\newcommand{\Poset}{\catname{Poset}}

\newcommand{\id}{\mathrm{id}}
\newcommand{\Id}{\mathrm{Id}}
\newcommand{\compose}{\circ}
\newcommand{\fcomp}{\cdot}
\newcommand{\natto}{\Rightarrow}
\newcommand{\op}{\mathrm{op}}
\newcommand{\iso}{\cong}

\newcommand{\0}{\mathbf 0}
\newcommand{\1}{\mathbf 1}
\newcommand{\tensor}{\otimes}
\newcommand{\lolly}{\multimap}
\newcommand{\lollyR}{\multimap^\mathrm{R}}
\newcommand{\lact}{\triangleright}
\newcommand{\ract}{\triangleleft}
\newcommand{\lactC}{\lact_\C}
\newcommand{\ractC}{\ract_\C}
\newcommand{\lactD}{\lact_\D}
\newcommand{\ractD}{\ract_\D}
\newcommand{\act}{\lact}
\newcommand{\actC}{\lactC}
\newcommand{\actD}{\lactD}
\newcommand{\power}{\leftpitchfork}
\newcommand{\powerC}{\power_\C}
\newcommand{\powerD}{\power_\D}
\renewcommand{\hom}{\rightarrowtriangle}

\newcommand{\jj}{j}
\newcommand{\jY}{\hat{\jmath}}
\newcommand{\iR}{i}
\newcommand{\pR}{p}
\newcommand{\copower}{\bullet}
\newcommand{\coproj}{\mathrm{in}}
\newcommand{\inl}{\mathrm{inl}}
\newcommand{\inr}{\mathrm{inr}}
\newcommand{\tuple}[1]{\langle#1\rangle}
\newcommand{\cotuple}[1]{[#1]}

\newcommand{\ctx}[1]{^{(#1)}}
\newcommand{\under}[1]{\underline{#1}}
\newcommand{\str}{\mathrm{str}}
\newcommand{\lstr}{\mathrm{str}}
\newcommand{\rstr}{\mathrm{str}^{\mathrm{R}}}

\newcommand{\pointfun}[1]{\llparenthesis #1 \rrparenthesis}

\newcommand{\List}{\mathrm{List}}
\newcommand{\ListM}{\mathsf{List}}
\newcommand{\append}{\mathbin{+\mkern-5mu+}}
\newcommand{\Wr}[1]{\mathrm{Wr}_{#1}}
\newcommand{\WrM}[1]{\mathsf{Wr}_{#1}}

\newcommand{\monoid}[1]{\mathsf{#1}}
\newcommand{\M}{\monoid M}
\newcommand{\Act}[1]{\catname{Act}\,#1}
\newcommand{\iact}{*}

\newcommand{\monad}[1]{\mathsf{#1}}
\renewcommand{\S}{\monad S}
\newcommand{\T}{\monad T}
\newcommand{\extend}[1]{#1^\dagger}
\newcommand{\pextend}[1]{#1^\dagger}
\newcommand{\Kl}[1]{\catname{Kl}\,#1}
\newcommand{\FAlg}[1]{\catname{alg}\,#1}
\newcommand{\TAlg}[1]{\catname{Alg}\,#1}
\newcommand{\fmap}{\mathrm{fmap}}
\newcommand{\bind}{\mathrm{bind}}

\newcommand{\sem}[1]{\llbracket #1 \rrbracket}

\begin{document}


\maketitle

\begin{abstract}
  Strong monads are important for several applications, in particular,
  in the denotational semantics of effectful languages, where strength
  is needed to sequence computations that have free variables.
  Strength is non-trivial: it can be difficult to determine whether a
  monad has any strength at all, and monads can be strong in multiple
  ways.
  We therefore review some of the most important known facts about strength 
and prove some new ones.
  In particular, we present a number of equivalent characterizations of strong functor
  and strong monad, and give some conditions that guarantee existence or uniqueness of strengths.
  We look at strength from three different perspectives: actions of a
  monoidal category $\V$, enrichment over $\V$, and powering over $\V$.
  We are primarily motivated by semantics of effects, but the results are
  also useful in other contexts.
\end{abstract}

\section{Introduction}
Following Moggi~\cite{moggi1989computational}, effectful computations are often modelled using
strong monads.
Strength also appears in other applications; for example, strength
is crucial for the notion of commutative monad~\cite{kock1970monads} used in the
construction of tensor products on categories of
algebras~\cite{linton1969coequalizers,kock1971bilinearity}, and
in measure theory~\cite{kock2012commutative}; strong functors are also important in the
study of abstract syntax~\cite{fiore1999abstract}.
It can be difficult in these contexts to determine whether a given functor or
monad admits a strength,
and various facts have been proved about strength to help with this.
Some appear in published work (often as a small lemma in a paper not
primarily about strength)~\cite{mulry2013notions,sato2018giry,levy2019strong}, while others are folklore.
These have some overlap, and levels of
generality vary.

We collect together a number of important results about strength.
There are two groups of results in particular that we focus on.
One is the equivalence of various definitions of strong functor and
strong monad.
These are useful in particular for reasoning about strong functors and
monads, and are also useful for constructing strengths for ordinary
functors and monads.
The other is results concerning existence and uniqueness of strengths
for functors and monads.
Several of these results are known, but a good number are, to the best
of our knowledge, new.

The difference between monads and strong monads is best seen by looking
at the \emph{Kleisli extension} operator.
If $T$ is the underlying endofunctor of a monad, then 
every morphism $f : X \to TY$ induces a morphism $\extend f : TX \to TY$, as
on the left below.
In the Cartesian case, if $T$ forms a \emph{strong} monad, then
the Kleisli extension has the more general form on the right.
\[
  \frac
    {f : X \to TY}
    {\extend f : TX \to TY}
  \qquad
  \frac
    {f : \Gamma \times X \to TY}
    {\extend f : \Gamma \times TX \to TY}
\]
Our main interest is the semantics of effects (though the results we give
here can be applied more widely).
Strength in this case enables interpretation of terms with free
variables.
Consider the following typing rule:
\[
  \frac
    {\Gamma \vdash t : A \qquad \Gamma, x : A \vdash t' : B}
    {\Gamma \vdash \mathsf{let}~x = t~\mathsf{in}~t' : B}
\]
In a monadic model of a call-by-value language, the terms $t$ and $t'$ would be interpreted
as morphisms $\sem{t} : \sem{\Gamma} \to T\sem{A}$ and
$\sem{t'} : \sem{\Gamma} \times \sem{A} \to T\sem{B}$, where $TX$ is the
object of (possibly effectful) computations that return values in $X$.
Using the strong Kleisli extension of $\sem{t'}$ we can interpret the
$\mathsf{let}$ as
\[
  \sem{\mathsf{let}~x = t~\mathsf{in}~t'} :
  \sem{\Gamma}
  \xrightarrow{\tuple{\id_{\sem{\Gamma}}, \sem{t}}}
  \sem{\Gamma} \times T\sem{A}
  \xrightarrow{\extend{\sem{t'}}}
  T\sem{B}
\]
The Kleisli extension of an ordinary monad suffices when $\Gamma$ is
empty (because then $\sem{\Gamma} = \1$), but
we need the strong version in general.

Instead of assuming products, we work in the more general setting of
an \emph{action} of a monoidal category on another category.
Strengths with respect to an action appear for example
in~\cite{cockett1992strong,fiore2008second,mellies2012parametric,kammar2017monad,szlachanyi2017tensor}.
Working with actions instead of a Cartesian, symmetric monoidal or
general monoidal structure does not add much complexity, but is useful
for some of the results we give.  We also approach strength from two
other perspectives.  The \emph{enriched} perspective is well-known for
categories enriched over themselves and goes back to
Kock~\cite{kock1972strong}; by generalizing to actions, we remove the
self-enrichment restriction.  The third perspective, which we call
\emph{powering}, is less well-known, but was also first considered by
Kock~\cite{Koc:clocgc}.
The same three-perspective approach can be found in the nLab article on
strong monads~\cite{ncatlab:strm}, but for the most part still only for
the self-enriched case.

We discuss actions, strong functors, and strong monads in
\cref{sec:actions,sec:strong-functors,sec:strong-monads}, looking
especially at uniqueness and existence of strengths for
functors.
Our novel contributions are sufficient criteria for unique existence
(based on our notion of \emph{functional completeness}), and for non-unique
existence (based on our notion of \emph{weak functional completeness}).
We also provide a number of examples. We consider enrichment
in \cref{sec:enrichment} and powering in \cref{sec:powering}. 
In \cref{sec:biactions}, we
discuss biactions, bistrong functors and commutative monads.


\section{Monoidal categories and actions}
\label{sec:actions}

We begin by recalling the notions of \emph{monoidal category} and
\emph{action}, and give our primary examples of these.

\begin{definition}
  A \emph{monoidal category} $(\V, I, \tensor)$ consists of a category
  $\V$, an object $I \in \V$ called the \emph{unit}, and a functor
  $\tensor : \V \times \V \to \V$ called the \emph{tensor},
  equipped with three natural isomorphisms
  \[
    \lambda_\Gamma : I \tensor \Gamma \to \Gamma
    \qquad
    \rho_\Gamma : \Gamma \to \Gamma \tensor I
    \qquad
    \alpha_{\Gamma_1,\Gamma_2,\Gamma_3}
      : (\Gamma_1 \tensor \Gamma_2) \tensor \Gamma_3
        \to \Gamma_1 \tensor (\Gamma_2 \tensor \Gamma_3)
  \]
      satisfying the following coherence conditions:
      \[
        \begin{tikzcd}[column sep=normal,ampersand replacement=\&]
          \Gamma \tensor \Delta \arrow[r, equals] \arrow[d, "\rho_\Gamma \tensor \Delta"'] \&
          \Gamma \tensor \Delta \\
          (\Gamma \tensor I) \tensor \Delta \arrow[r, "\alpha_{\Gamma, I, \Delta}"'] \&
          \Gamma \tensor (I \tensor \Delta) \arrow[u, "\Gamma \tensor \lambda_\Delta"']
        \end{tikzcd}
        \hspace{-1.4em}
        \begin{tikzcd}[column sep=-4.2em, row sep=1.6em,ampersand replacement=\&]
          \&\&
          (\Gamma_1 \tensor \Gamma_2) \tensor (\Gamma_3 \tensor \Gamma_4)
          \arrow[rrd, "\alpha_{\Gamma_1, \Gamma_2, \Gamma_3\tensor \Gamma_4}"]
          \\
          \mathrlap{\hspace{1.6em}((\Gamma_1 \tensor \Gamma_2) \tensor \Gamma_3) \tensor \Gamma_4}
          \phantom{((\Gamma_1 \tensor \Gamma_2) \tensor \Gamma_3) \tensor \Gamma_4}
          \arrow[rru, "\alpha_{\Gamma_1 \tensor \Gamma_2, \Gamma_3, \Gamma_4}"]
          \arrow[rd, "\alpha_{\Gamma_1, \Gamma_2, \Gamma_3} \tensor \Gamma_4"'] \&\&\&\&
          \phantom{\Gamma_1 \tensor (\Gamma_2 \tensor (\Gamma_3 \tensor \Gamma_4))}
          \mathllap{\Gamma_1 \tensor (\Gamma_2 \tensor (\Gamma_3 \tensor \Gamma_4))\hspace{1.6em}} \\
          \&
          (\Gamma_1 \tensor (\Gamma_2 \tensor \Gamma_3)) \tensor \Gamma_4
          \arrow[rr, "\alpha_{\Gamma_1, \Gamma_2\tensor \Gamma_3, \Gamma_4}"'{yshift=-1ex}]
          \&\&
          \Gamma_1 \tensor ((\Gamma_2 \tensor \Gamma_3) \tensor \Gamma_4)
          \arrow[ru, "\Gamma_1 \tensor \alpha_{\Gamma_2, \Gamma_3, \Gamma_4}"']
        \end{tikzcd}
      \]
\end{definition}

\begin{example}
Every category $\V$ with finite products forms a
\emph{Cartesian monoidal category} $(\V, \1, \times)$, in which the unit
is the terminal object $\1$, and the tensor of $X$ and $Y$ is the binary
product $X \times Y$.
\end{example}

\begin{example}
  Let $\PSet$ be the category of pointed sets and point-preserving
  functions.
  Objects of $\PSet$ are sets $X$ with a distinguished element
  $\star \in X$; morphisms are functions $f : X \to Y$ such that
  $f \star = \star$.
  We consider two monoidal structures on $\PSet$.
  The first is the Cartesian monoidal structure, which is inherited from
  $\Set$ (the product $X \times Y$ is the product of sets, with
  distinguished element $(\star, \star)$).
  The second is the \emph{smash product}
  $X \tensor Y = \{(x, y) \in X \times Y \mid x = \star \Leftrightarrow y = \star\}$,
  which has the two-element pointed set $\{\star, 1\}$ as the unit.
  On morphisms, $\tensor$ is given by
  $(f \tensor g) (x, y) = (\star, \star)$ if $f x = \star$ or
  $g y = \star$, and by $(f \tensor g) (x, y) = (fx, gy)$ otherwise.
\end{example}

\begin{example}
  Let $\M = (M, 1, \cdot)$ be a (set-theoretic) monoid.
  The category $\Act \M$ of right $\M$-actions has as objects sets $X$
  equipped with a function $(*) : X \times M \to X$ such that
  $x * 1 = x$ and $x * (m \cdot m') = (x * m) * m'$ for all $x \in X$
  and $m, m' \in M$.
  Morphisms $f : X \to Y$ in $\Act \M$ are functions that preserve the
  action, i.e.\ $f(x * m) = (f x) * m$ for all $x \in X$ and $m \in M$.
  The category $\Act \M$ is Cartesian monoidal; the terminal object $\1$ is the
  one-element set equipped with the unique $*$, and the product
  $X \times Y$ is the product of sets with
  $(x, y) * m = (x * m, y * m)$.

  When $\M$ is natural numbers with addition, $\Act \M$ is isomorphic to
  the category of sets $X$ equipped with an endofunction $e : X \to X$;
  morphisms $f : X \to Y$ are functions such that
  $f \compose e = e \compose f$.
  The action on an object $X$ is $x * n = e^n\,x$. This is isomorphic to the category $[\mathbb{N}, \Set]$ where $\mathbb{N}$ is the one-object category with natural numbers as morphisms and addition as composition. 
\end{example}

\begin{definition}
  A \emph{(left) action}\footnote{A category $\C$ with a left action of a monoidal category $\V$ is also called a (left) $\V$-\emph{actegory}.} of a monoidal category
  $(\V, I, \tensor)$ on a category $\C$ is a functor
  $\act : \V \times \C \to \C$ equipped with two natural isomorphisms
  \[
    \lambda_X : I \act X \to X
    \qquad
    \alpha_{\Gamma',\Gamma,X}
      : (\Gamma' \tensor \Gamma) \act X
        \to \Gamma' \act (\Gamma \act X)
  \]
      satisfying the following coherence conditions:
      \[
\begin{tikzcd}[ampersand replacement=\&]
	{(I \tensor \Gamma) \lact X} \& {I \lact (\Gamma \lact X)} \\
	{\Gamma \lact X} \& {\Gamma \lact X} \\
	{(\Gamma \tensor I) \lact X} \& {\Gamma \lact (I \lact X)}
	\arrow["{\lambda_{\Gamma \lact X}}", from=1-2, to=2-2]
	\arrow["{\Gamma \lact \lambda_X}"', from=3-2, to=2-2]
	\arrow["{\alpha_{\Gamma, I, X}}"', from=3-1, to=3-2]
	\arrow["{\rho_\Gamma \lact X}"', from=2-1, to=3-1]
	\arrow["{\lambda_\Gamma \lact X}"', from=1-1, to=2-1]
	\arrow["{\alpha_{I, \Gamma, X}}", from=1-1, to=1-2]
	\arrow[Rightarrow, no head, from=2-1, to=2-2]
\end{tikzcd}
        \hspace{-1.0em}
        \begin{tikzcd}[column sep=-3.2em, row sep=1.6em, ampersand replacement=\&]
          \&\&
          (\Gamma_1 \tensor \Gamma_2) \act (\Gamma_3 \act X)
          \arrow[rrd, "\alpha_{\Gamma_1, \Gamma_2, \Gamma_3\act X}"]
          \\
          \mathrlap{\hspace{1.6em}((\Gamma_1 \tensor \Gamma_2) \tensor \Gamma_3) \act X}
          \phantom{((\Gamma_1 \tensor \Gamma_2) \tensor \Gamma_3) \act X}
          \arrow[rru, "\alpha_{\Gamma_1 \tensor \Gamma_2, \Gamma_3, X}"]
          \arrow[rd, "\alpha_{\Gamma_1, \Gamma_2, \Gamma_3} \act X"'] \&\&\&\&
          \phantom{\Gamma_1 \act (\Gamma_2 \act (\Gamma_3 \act X))}
          \mathllap{\Gamma_1 \act (\Gamma_2 \act (\Gamma_3 \act X))\hspace{1.6em}} \\
          \&
          (\Gamma_1 \tensor (\Gamma_2 \tensor \Gamma_3)) \act X
          \arrow[rr, "\alpha_{\Gamma_1, \Gamma_2\tensor \Gamma_3, X}"'{yshift=-0.4ex}]
          \&\&
          \Gamma_1 \act ((\Gamma_2 \tensor \Gamma_3) \act X)
          \arrow[ru, "\Gamma_1 \act \alpha_{\Gamma_2, \Gamma_3, X}"']
        \end{tikzcd}
      \]
\end{definition}

A left action of $\V$ on $\C$ is the same as a monoidal functor
from $\V$ to $[\C,\C]$, where we equip $[\C,\C]$ with the composition monoidal structure. 


\begin{example}
  The tensor of any monoidal category $\V$ (in particular,
  the examples above) forms an action of $\V$ on
  itself, with $X \act Y = X \tensor Y$.
\end{example}

\begin{example}\label{example:copowers}
  Consider $\V = \Set$ with the Cartesian monoidal
  structure.
  A category $\C$ \emph{has copowers} over $\Set$ when for all sets $\Gamma$ and
  objects $X \in \C$, the coproduct
  $\Gamma \copower X = \coprod_{\gamma \in \Gamma} X \in \V$ exists.
  The object $\Gamma \copower X$ is the \emph{copower} of $\Gamma$ and
  $X$; its universal property is that morphisms
  $f : \Gamma \copower X \to Y$ are in natural bijection with tuples
  $(f_\gamma : X \to Y)_{\gamma \in \Gamma}$ of morphisms, by taking
  $f_\gamma = f \compose \coproj_{\gamma}$.
  If $\C$ has copowers over $\Set$, then they form an action
  $\Gamma \act X = \Gamma \copower X$ of $(\Set, \1, \times)$ on $\C$.
  For $\C = \Set$, the copower $X \copower Y$ is just the Cartesian
  product $X \times Y$.
\end{example}

In the relationship between strength and enrichment explained below in
\cref{sec:enrichment}, the action $\act$ forms the \emph{copowers} (or
\emph{tensors}) of the enriched category in a more general sense of
`copower'. Any locally small category $\C$ is uniquely
$\Set$-enriched. Its copowers, if they exist, are given by small
coproducts as described above.

\section{Strong functors}
\label{sec:strong-functors}

Throughout this section, we suppose a monoidal category
$\V$, whose objects $\Gamma$ we view as \emph{contexts}
(because of their role in the introduction as interpretations of typing
contexts).
We then consider strong functors $F : \C \to \D$, where $\C$ and $\D$
are categories equipped with actions $\actC : \V \times \C \to \C$ and
$\actD : \V \times \D \to \D$.
We have no need to assume that $\V$ is a \emph{symmetric} monoidal category (but
this is the case for all of our examples).
A $\C$-morphism $\Gamma \actC X \to Y$ can be thought of as a morphism
from $X$ to $Y$ \emph{in context $\Gamma$}, and similarly for $\D$.

There are several equivalent definitions of strong functor.
The following is not the standard one, but matches closely the intuition that the context
$\Gamma$ should be preserved, enabling the interpretation of terms with
free variables.
\begin{definition}
  A \emph{(left) strong functor} $F : \C \to \D$ consists of an object
  $FX \in \D$ for each object $X \in \C$ and a $\D$-morphism
  $F\ctx{\Gamma} f : \Gamma \actD FX \to FY$ for each $\C$-morphism
  $f : \Gamma \actC X \to Y$, such that $F\ctx{\Gamma}$ is natural
  in $\Gamma \in \V$, 
and 
\[
\begin{array}{@{}c@{\quad}l}
  F\ctx{I}\lambda_X = \lambda_{FX} 
   &\textrm{for $X \in \C$}
\\
    F\ctx{\Gamma'\tensor\Gamma}
      (g \compose (\Gamma' \actC f) \compose \alpha_{\Gamma',\Gamma,X})
    = F\ctx{\Gamma'} g \compose (\Gamma' \actD F\ctx{\Gamma}f)
      \compose \alpha_{\Gamma',\Gamma, FX}
  &\textrm{for $f : \Gamma \actC X \to Y$, $g : \Gamma' \actC Y \to Z$}
\end{array}
\]
  If $F, G : \C \to \D$ are strong functors, then a
  \emph{strong natural transformation} $\tau : F \natto G$ consists
  of a morphism $\tau_X : FX \to GX$ for each $X \in \C$, such that
  $
    \tau_Y \compose F\ctx{\Gamma} f
    = G\ctx{\Gamma} f \compose (\Gamma \actD \tau_X)
  $
  for each $f : \Gamma \actC X \to Y$.
\end{definition}

Every strong functor $F$ has an underlying ordinary functor
$\under F : \C \to \D$, given on objects by $\under F X = FX$
and on morphisms $f : X \to Y$ by
$
  \under F f = F\ctx{I} (f \compose \lambda_X) \compose \lambda^{-1}_{FX}
    : FX \to FY
$.
Every strong natural transformation $\alpha : F \natto G$
is a natural transformation $\alpha : \under F \natto \under G$.
There is an identity strong functor $\Id$ and each pair of strong
functors $F : \C \to \D$ and $G : \D \to \E$ has a composition
$G \fcomp F : \C \to \E$. These are given on objects
$X$ and morphisms $f : \Gamma \actC X \to Y$ by
\[
  \Id X = X
  \quad
  \Id\ctx{\Gamma} f = f
  \qquad\qquad
  (G \fcomp F) X = G (F X)
  \quad
  (G \fcomp F)\ctx{\Gamma} f = G\ctx{\Gamma} (F\ctx{\Gamma} f)
\]

\begin{example}\label{example:list-functor}
  Let $\actC$ and $\actD$ both be the action of $\Set$ on itself
  given by the Cartesian monoidal structure.
  The strong functor $\List : \Set \to \Set$ maps each set $X$ to the
  set of lists over $X$; on functions $f : \Gamma \times X \to Y$ it is
  given by
$
    \List\ctx{\Gamma} f\,(\gamma, [x_1, \dots, x_n]) =
      [f (\gamma, x_1), \dots, f (\gamma, x_n)]
$.
  The ordinary functor $\under{\List} : \Set \to \Set$ is then just the
  usual list functor, given on functions by
  $\under{\List}\,f\,[x_1, \dots, x_n] = [f x_1, \dots, f x_n]$.
\end{example}

An alternative definition is that a strong functor is an ordinary functor
equipped with a \emph{strength}.
This is the more common definition, and we make some use of it below.
\begin{definition}
  A \emph{(left) strength} for an ordinary functor $F : \C \to \D$
  is a family of $\D$-morphisms
  $\str_{\Gamma, X} : \Gamma \actD FX \to F(\Gamma \actC X)$, natural in
  $\Gamma \in \V$ and $X \in \C$ and such that
  \[
\begin{tikzcd}
	{I \actD FX} \\
	{F(I \actC X)} & FX
	\arrow["{\str_{I, X}}"', from=1-1, to=2-1]
	\arrow["{F\lambda_X}"', from=2-1, to=2-2]
	\arrow["{\lambda_{FX}}", from=1-1, to=2-2]
\end{tikzcd}
    ~~
\begin{tikzcd}
	{(\Gamma' \tensor \Gamma) \actD FX} & {\Gamma' \actD (\Gamma \actD FX)} & {\Gamma' \actD F(\Gamma \actC X)} \\
	{F((\Gamma' \tensor \Gamma) \actC X)} && {F(\Gamma' \actC (\Gamma \actC X))}
	\arrow["\alpha_{\Gamma',\Gamma,FX}", from=1-1, to=1-2]
	\arrow["{\Gamma' \actD \str_{\Gamma, X}}", from=1-2, to=1-3]
	\arrow["{\str_{\Gamma', \Gamma \actC X}}", from=1-3, to=2-3]
	\arrow["{\str_{\Gamma' \tensor \Gamma, X}}"', from=1-1, to=2-1]
	\arrow["F\alpha_{\Gamma',\Gamma,X}"', from=2-1, to=2-3]
\end{tikzcd}
  \]
\end{definition}

We show the equivalence between these two definitions of strong
functor and the corresponding fact for strong natural transformations.

\begin{proposition}\label{prop:strong-functors-strength}
  If $F : \C \to \D$ is an ordinary functor, then there is a bijection
  between (1) strong functors $\hat F$ such that $\under{\hat F} = F$,
  and (2) strengths $\str$ for $F$.
  If $F, G : \C \to \D$ are functors equipped with the equivalent
  data of this bijection, then a
  natural transformation $\tau : F \natto G$ is a strong natural
  transformation $\tau : \hat F \natto \hat G$ exactly when the
  following commutes:
  \[
    \begin{tikzcd}[column sep=large, row sep=small]
      \Gamma \actD FX
      \arrow[r, "\Gamma \actD \tau_X"]
      \arrow[d, "\str_{\Gamma, X}"'] &
      \Gamma \actD GX
      \arrow[d, "\str_{\Gamma, X}"] \\
      F(\Gamma \actC X)
      \arrow[r, "\tau_{\Gamma \actC X}"'] &
      G(\Gamma \actC X)
    \end{tikzcd}
  \]
\end{proposition}
\begin{proof}
  By the Yoneda lemma, families of functions
$    \hat F\ctx{\Gamma}
      : \C(\Gamma \actC X, Y) \to \D(\Gamma \actD FX, FY)
$
  natural in $Y \in \C$ are in bijection with morphisms
  $\str_{\Gamma, X} : \Gamma \actD FX \to F(\Gamma \actC X)$.
  If $\hat F$ is a strong functor with $\under{\hat F} = F$,
  then $\hat F\ctx{\Gamma}$ is natural in $Y$.
  Moreover, such a natural family 
  forms a strong functor $\hat F$ with $\under{\hat F} = F$ exactly when
  $\str_{\Gamma, X} = \hat F\ctx{\Gamma}\,\id_{\Gamma\actC X}$
  is a strength for $F$.
  The fact about natural transformations
  follows immediately from the fact that given a strength $\str$, the
  corresponding strong functor is given by
  $\hat F\ctx{\Gamma} f = F f \compose \str_{\Gamma, X}$.
\end{proof}

\subsection{Uniqueness and existence of strengths}\label{sec:uniqueness}
It is well-known that every endofunctor $F$ on $\Set$ has a unique
strength with respect to the Cartesian monoidal structure.
Various other results about uniqueness of strengths have been proved
(e.g.\ \cite{moggi1989computational,sato2018giry,levy2019strong}).
Uniqueness results are useful for determining whether a given functor
admits a strength at all when for some reason there is only one candidate to check.
Conversely, existence results make it easier to construct strengths for
functors.
We supply uniqueness and existence results for strengths in this
section.

We first define a notion of \emph{functional completeness} for an
action, which guarantees both uniqueness and existence of strengths.%
\footnote{This notion is similar in spirit to functional completeness in
categorical logic~\cite{Lam:func,Pav:catlna}, but not the same. 
In categorical logic, $\V$ would be functionally complete if, 
for any $\V[\Gamma]$-morphism $z : X \to Y$, there were a unique $\V$-morphism $f : \Gamma \otimes X \to Y$ such that $J f \compose (\gamma \otimes X) \compose \lambda^{-1}_X = z$.
Here by $\V[\Gamma]$ we mean the monoidal category obtained by freely extending $\V$ with a morphism $\gamma : I \to \Gamma$---an ``indeterminate'' point of $\Gamma$---and by $J$ we mean the inclusion of $\V$ in $\V[\Gamma]$.
(For a Cartesian monoidal $\V$, one would extend to a Cartesian monoidal category $\V[\Gamma]$ instead of just a monoidal category.)
We are looking for a more distinctive name for our notion.}
We call the elements $\gamma \in \V(I, \Gamma)$ the \emph{points} of
$\Gamma$.
For each morphism $f : \Gamma \actD X \to Y$ in $\D$, we have a function
$\pointfun f : \V(I, \Gamma) \to \D(X, Y)$, \emph{applying} $f$ to
points, by defining
$\pointfun f \gamma = f \compose (\gamma\actD X)\compose \lambda^{-1}_X$.
\begin{definition}
  We say that $\actD$ is \emph{functionally complete} if, for every
  function $\zeta : \V(I, \Gamma) \to \D(X, Y)$, there is a unique
  $\D$-morphism $f : \Gamma \actD X \to Y$ such that $\pointfun {f} = \zeta$.
\end{definition}

Writing $\Phi_\Gamma \zeta$ for the unique $f$ from the definition, we
get a family of functions $\Phi$ that are the inverses of the functions
$\pointfun{-}$. This family $\Phi$ is natural in $\Gamma$, $X$ and
$Y$. Moreover, it satisfies
\[
\begin{array}{@{}c@{\qquad}l}
    \Phi_{I} \zeta 
    = \zeta \id_I \compose \lambda_X
    & \textrm{for~} \zeta : \V(I, I) \to \D(X, Y)
    \\
    \Phi_{\Gamma' \tensor \Gamma} \zeta
    =
    \Phi_{\Gamma'} (\lambda \gamma'.\,
      \Phi_{\Gamma} (\lambda \gamma.\
        \zeta ((\gamma' \tensor \gamma) \compose \rho_I)))
          \compose \alpha_{\Gamma', \Gamma, X}
    & \textrm{for~} \zeta : \V(I, \Gamma' \tensor \Gamma) \to \D(X, Y)
\end{array}
\]
The key consequence is the following (which is a corollary of
\cref{prop:strengths-unique,prop:strengths-exist} below):
\begin{proposition}
  If $\actD$ is functionally complete, then every functor
  $F : \C \to \D$ has a unique strength, and every natural
  transformation $\tau : F \natto G$ of functors $\C \to \D$ is
  strong.
\end{proposition}

\begin{example}
  The category $\Set$ with the Cartesian product is functionally
  complete; this is why endofunctors on $\Set$ have unique strengths.
  More generally, if $\D$ has copowers over $\Set$, then the action $\copower$ of
  $(\Set, \1, \times)$ on $\D$ is functionally complete.
  In this case, points $\gamma \in \Set(\1, \Gamma)$ are just elements
  of the set $\Gamma$.
  Functional completeness thus says equivalently that, for every
  function $\zeta : \Gamma \to \D(X, Y)$, there is a unique
  $f : \Gamma \copower X \to Y$ such that
  $f \compose \coproj_{\gamma} = \zeta \gamma$ for
  all $\gamma \in \Gamma$. This is exactly the universal property of
  $\Gamma \copower X$.
\end{example}

\begin{example}
  In contrast, the Cartesian product of pointed sets is not functionally
  complete as an action of $(\PSet, \1, \times)$ on itself: every $\Gamma$
  has only one point, so the morphisms $f : \Gamma \times X \to Y$
  fail to be unique.
  The Cartesian product of posets also fails to be functionally complete as an action of $(\Poset, \1, \times)$ on itself: the morphisms
  $f : \Gamma \times X \to Y$ in this case are necessarily given by
  $f(\gamma, x) = \zeta(\gamma)(x)$, so are unique if they exist, but this $f$ may fail to be
  monotone.
\end{example}

There are few examples of functionally complete actions.
We break down the notion of functional completeness into
\emph{well-pointedness}, which guarantees uniqueness of strength, and
the existence of a \emph{weak functional completeness structure}, which
guarantees existence of a canonical strength for each functor.
Both of these have more examples.

\begin{definition}
  The action $\actD$ is said to be \emph{well-pointed} (or \emph{have enough points}) when $\pointfun {{-}}$ is
  injective, i.e.\ when $\pointfun f = \pointfun g$ implies $f = g$ for
  all $f, g : \Gamma \actD X \to Y$.
\end{definition}
If the action is a monoidal category acting on itself,
well-pointedness in our sense is equivalent to Abramsky and Heunen's
\emph{(monoidal) well-pointedness}~\cite{abramsky2012h}.
According to their
definition, a monoidal category $\V$ is well-pointed if two morphisms
$f, g : X \otimes X' \to Y$ are equal whenever
$f \compose (\xi \otimes \xi') \compose \rho_I = g \compose (\xi
\otimes \xi') \compose \rho_I$
for all $\xi : I \to X$, $\xi' : I \to X'$.
For a Cartesian monoidal category acting on itself, our
notion of well-pointedness agrees with the usual notion defined with respect
to a terminal object (two morphisms $f, g : X \to Y$ are equal if
$f \compose \xi = g \compose \xi$ for all $\xi : \1 \to X$).
A similar simplification is possible when each ${-}\actD X$ has a
right adjoint (i.e.\ there is a corresponding
\emph{enrichment} in the sense of \cref{sec:enrichment}, for example
when $\actD$ is the action of a monoidal closed category on itself).
When these right adjoints exist, if two $\V$-morphisms
$f, g : \Gamma \to \Delta$ are equal whenever
$f \compose \gamma = g \compose \gamma$ for all $\gamma : I \to \Gamma$, then
$\actD$ is well-pointed.

Functional completeness is a strictly stronger property than
well-pointedness; the latter implies that strengths are unique if they
exist, but does not guarantee existence.

\begin{proposition}\label{prop:strengths-unique}
  Suppose that $\actD$ is well-pointed.
  A functor $F : \C \to \D$ has a strength exactly when, for every
  $\Gamma, X$, there is a morphism
  $\str_{\Gamma, X} : \Gamma \actD F X \to F(\Gamma \actC X)$ such that
  \begin{equation}
    \pointfun{\str_{\Gamma, X}} \gamma
      = F(\pointfun{\id_{\Gamma \actC X}} \gamma)
      \quad \textrm{for~} \gamma \in \V(I, \Gamma)
  \end{equation}
  When the (necessarily unique) morphisms $\str_{\Gamma, X}$ exist,
  $\str$ is the only strength for $F$.
  Moreover, if $G, H : \C \to \D$ are strong functors, then every
  natural transformation $\tau : G \natto H$ is strong.
\end{proposition}
\begin{proof}
  If a family of morphisms $\str$ is a strength, then
  \[
    \pointfun{\str_{\Gamma, X}}\gamma
    =
    \str_{\Gamma, X} \compose (\gamma \act_\D F X) \compose \lambda^{-1}_{FX}
    = 
    F(\gamma \act_{\C} X) \compose \str_{I, X}
      \compose \lambda^{-1}_{FX}
    =
    F(\gamma \act_{\C} X) \compose F \lambda^{-1}_X
    = F(\pointfun{\id_{\Gamma \actC X}} \gamma)
  \]
  using naturality in $\Gamma$ and the strength axiom for $\lambda$.
  Well-pointedness therefore implies uniqueness.
  To show that every family of morphisms $\str$ satisfying
  the condition 
  is a strength, it suffices by
  well-pointedness to consider the image of each axiom under
  $\pointfun{{-}}$
 and then calculate.  
  A similar proof shows that every $\tau$ is strong.
\end{proof}

\begin{example}
  The smash product $\tensor$ of pointed sets is a well-pointed action of
  $(\PSet, I, \tensor)$ on $\PSet$.
  In this case, the unit $I$ is the two-element pointed set
  $\{\star, 1\}$; since morphisms $I \to \Gamma$ send $\star$ to
  $\star$, points of $\Gamma$ are in bijection with elements of
  $\Gamma$, and so, for any $\zeta$, there is at most one $f$ with
  $\pointfun f = \zeta$.
  In contrast, the Cartesian product of pointed sets is not
  well-pointed, because every $\Gamma$ has only one point.

  The Cartesian product of posets is a well-pointed action of
  $(\Poset, \1, \times)$ on itself, again because points of $\Gamma$ are
  in bijection with elements of $\Gamma$.
  If $F$ is an endofunctor on $\Poset$, then any strength for $F$ would
  necessarily be given by
  $\str_{\Gamma, X} (\gamma, t) = F (\lambda x.\,(\gamma, x)) t$;
  this forms a strength for $F$ exactly when it is monotone.
  This property enables us to show that some functors have no strength.
  For example, the functor $|{-}| : \Poset \to \Poset$ that
  sends $(X, \le)$ to the discrete poset $(X, =)$ has no
  strength because $\str_{\Gamma, X}$
  is not a monotone function $\Gamma \times |X| \to |\Gamma \times X|$.
\end{example}

For existence of strengths, we introduce the following.
\begin{definition}
  A \emph{weak functional completeness structure} $\Phi$ for $\actD$ is
  an assignment of a $\D$-morphism
  $\Phi_\Gamma \zeta : \Gamma \act X \to Y$
  satisfying $\pointfun{\Phi_\Gamma \zeta} = \zeta$
  to each function $\zeta : \V(I, \Gamma) \to \D(X, Y)$.
  We require this assignment to be natural in $\Gamma$, $X$, $Y$, and to satisfy
\[
    \Phi_{\Gamma' \tensor \Gamma} \zeta
    =
    \Phi_{\Gamma'} (\lambda \gamma'.\,
      \Phi_{\Gamma} (\lambda \gamma.\
        \zeta ((\gamma' \tensor \gamma) \compose \rho_I)))
          \compose \alpha_{\Gamma', \Gamma, X}
    \qquad \textrm{for~} \zeta : \V(I, \Gamma' \tensor \Gamma) \to \D(X, Y)
\]
\end{definition}

If $\actD$ is well-pointed, then there can be at most one weak functional
completeness structure $\Phi$ for $\actD$, because morphisms $f$
such that $\pointfun{f} = \zeta$ are unique if they
exist;
such a $\Phi$ exists exactly when $\actD$ is functionally complete.
If $\actD$ is not well-pointed, then in general there can be several weak
functional completeness structures for $\actD$.

\begin{proposition}\label{prop:strengths-exist}
  Let $\Phi$ be a weak functional completeness structure for the action $\actD$.
  For every functor $F : \C \to \D$, there is a strong functor
  $\hat F$ such that $\under{\hat F} = F$; this is defined on morphisms
  $f : \Gamma \actC X \to Y$ by
  $
    \hat F \ctx{\Gamma} f
      = \Phi_\Gamma(\lambda \gamma.\,F(\pointfun{f}\gamma))
  $.
  Moreover, every natural transformation
  $\alpha : F \natto G$ is a strong natural transformation
  $\hat F \natto \hat G$, for $\hat F$ and $\hat G$ thus constructed.
\end{proposition}

Weak functional completeness does not necessarily deliver the
canonical strengths for the identity functor and the composition of
functors with strength. 

\begin{example}
  We construct a weak functional completeness structure for
  $(\PSet, \1, \times)$ acting on itself.
  Every $\Gamma$ has exactly one point $\star$, so to give a function
  $\zeta : \PSet(\1, \Gamma) \to \PSet(X, Y)$ is just to choose a
  morphism $\zeta(\star) : X \to Y$ of pointed sets.
  We can therefore define the morphism $\Phi_\Gamma \zeta : \Gamma \times X \to Y$
  by $\Phi_\Gamma \zeta (\gamma, x) = \zeta(\star)(x)$.
  This is in fact the only possible $\Phi$ in this case, even though the action fails to be well-pointed.
  By \cref{prop:strengths-exist}, every endofunctor $F$ on pointed sets forms a strong
  functor with
  $
    \hat F \ctx \Gamma f\,(\gamma, t)
      = F (\lambda x.\,f(\star, x))\,t
  $; this corresponds to the strength
  $\str_{\Gamma,X} (\gamma, t) = F(\lambda x.\,(\star, x))\,t$.
  There may in general be other strengths for $F$.
  For example, the identity functor on $\PSet$ also has the canonical strength
  $\id_{\Gamma\times X} : \Gamma \times X \to \Gamma \times X$.
\end{example}

\begin{example}
  We give an example of an action that has multiple weak functional
  completeness structures.
  Fix a set $E$, and let $\D$ be the Kleisli category of the monad
  ${-} + E$ on $\Set$: objects are sets, and morphisms $f \in \D(X, Y)$
  are functions $f : X \to Y + E$; the identities are the left
  coprojections $\inl$, and the composition of
  $f : X \to Y + E$ with $g : Y \to Z + E$ is
  $\cotuple{g, \inr} \compose f : X \to Z + E$.
  This category is coCartesian, as is the Kleisli category of any monad on any coCartesian category. 
  It therefore forms a monoidal category $(\D, \0, +)$, which acts on
  itself.
  Every $\Gamma$ has exactly one point $\cotuple{}$ because $\0$ is initial, so a
  function $\zeta : \D(\0, \Gamma) \to \D(X, Y)$ just chooses a single
  function $\zeta{(\cotuple{})} : X \to Y + E$.
  For each $e \in E$, we therefore have a morphism
  $
    \Phi^e_\Gamma \zeta = \cotuple{\inr \compose e \compose \tuple{}, \zeta{(\cotuple{})}}
      \in \D(\Gamma + X, Y)
  $,
  and $\Phi^e$ is a weak functional completeness structure.
  Hence in general, there is more than one $\Phi$.
\end{example}

We note that, if $\D$ has copowers over $\Set$, then well-pointedness
of $\D$ amounts to the canonical morphisms
$\V(I,\Gamma) \bullet X \to \Gamma \act X$ being epimorphisms,
while weak functional completeness amounts to the monoidal natural
transformation $\V(I,{-}) \bullet ({=}) \to ({-})\act ({=})$ being a split
monomorphism (in the category of lax monoidal functors $\V \to [\D,\D]$).
Functional completeness is equivalent to the latter being an
isomorphism.

\section{Strong monads}
\label{sec:strong-monads}

We now turn to strong monads.
There is a richer collection of equivalent definitions of strong monad
than there is of strong functor.
Because of our focus on semantics, the primary definition we use asks
for a strong Kleisli extension operator $\extend{({-})}$, as in the
introduction.
Again we work in the action-based setting, so we suppose a monoidal
category $\V$ that acts on a category $\C$.
We drop the subscript on the action, writing $\act$ instead of $\actC$.

\begin{definition}
  A \emph{strong monad} $\T = (T, \eta, \extend{({-})})$ consists of an
  object $TX \in \C$ and morphism $\eta_X : X \to TX$ for each
  $X \in \C$, and a morphism $\extend f : \Gamma \act TX \to TY$ for
  each $f : \Gamma \act X \to TY$, such that $\extend{({-})}$ is natural
  in $\Gamma \in \V$ and
\[
\begin{array}{@{}c@{\qquad}l}
    \extend{(\eta_X \compose \lambda_X)}
    = \lambda_{TX}
    & \textrm{for~} X \in \C
    \\
    \extend f \compose (\Gamma \act \eta_X)
    = f
    & \textrm{for~} f : \Gamma \act X \to TY
    \\
    \extend g \compose (\Gamma' \act \extend f)
      \compose \alpha_{\Gamma', \Gamma, TX}
    =
    \extend{(\extend g \compose (\Gamma' \act f)
      \compose \alpha_{\Gamma', \Gamma, X})}
    & \textrm{for~} f : \Gamma \act X \to TY$, $g : \Gamma' \act Y \to TZ
\end{array}
\]
  If $\S$ and $\T$ are strong monads, then a
  \emph{strong monad morphism} $\tau : \S \to \T$ consists of a
  morphism $\tau_X : SX \to TX$ for each $X \in \C$, such that
  $\tau_X \compose \eta_X = \eta_X$ for each $X \in \C$ and
  $
    \tau_Y \compose \extend{f}
      = \extend{(\tau_Y \compose f)} \compose (\Gamma \act \tau_X)
  $
  for each $f : X \to SY$.
\end{definition}

The morphisms $\eta$ are collectively called the unit of $\T$,
and $\extend{({-})}$ is the Kleisli extension.
If $\T = (T, \eta, \extend{({-})})$ is a strong monad, then the
assignment on objects $T$ extends to a strong functor with
$T\ctx{\Gamma} f = \extend{(\eta_Y \compose f)} : \Gamma \act TX \to TY$
for $f : \Gamma \act X \to Y$.
The unit $\eta$ is then a strong natural transformation
$\eta: \Id \natto T$, as is the \emph{multiplication}
$\mu : T \fcomp T \natto T$, given by
$\mu_X = \extend{\lambda_{TX}} \compose \lambda^{-1}_{TTX}$.
Every strong monad morphism $\tau : \S \to \T$ is a strong natural
transformation $\tau : S \natto T$. 

\begin{example}
  Consider $\Set$ with the Cartesian monoidal
  structure, acting on itself.
  The strong monad $\ListM$ on $\Set$ maps each set $X$ to the set
  $\List X$ of lists over $X$.
  The unit $\eta$ is given by the singleton lists $\eta_X x = [x]$.
  The Kleisli extension $\extend f : \Gamma \times \List X \to \List Y$
  of $f : \Gamma \times X \to \List Y$ is defined by
$
    \extend{f} (\gamma, [x_1, \dots, x_n])
      = f (\gamma, x_1) \append
        \cdots \append f (\gamma, x_n)
$
  where $\append$ is concatenation of lists.
  The strong functor $\List$ is the strong functor defined in
  \cref{example:list-functor} (and the ordinary functor
  $\under\List$ is then the usual list endofunctor on $\Set$).
\end{example}

If $\T = (T, \eta, \extend{({-})})$ is a strong monad on $\C$, then the
underlying functor $\under T$ forms an ordinary monad
$\under \T = (\under T, \eta, \mu)$ on $\C$ (where the multiplication
$\mu$ is defined as above).
Every strong monad morphism $\tau : \S \to \T$ is a monad morphism
$\tau : \under \S \to \under \T$.

We now give several equivalent characterizations of strong monads.
In addition to the definition above, strong monads can be defined in
terms of strong functors, in terms of strengths, and also by lifting the
action of $\V$ to the Kleisli category.
\begin{proposition}
  For each monad $\T = (T, \eta, \mu)$ on $\C$ there are bijections
  between
  \begin{enumerate}
    \item strong monads $\hat\T$ such that $\under{\hat\T} = \T$;
    \item strong functors $\hat T$ such that $\under{\hat T} = T$ and
      such that $\eta$ and $\mu$ are strong natural
      transformations $\Id \natto \hat T$ and $\hat T \fcomp \hat T \natto \hat T$;
    \item strengths $\str$ for the functor $T$, such that the following
      diagrams commute:
      \[
        \begin{tikzcd}
          \Gamma \act X
          \arrow[r, "\Gamma \act \eta_X"]
          \arrow[rd, "\eta_{\Gamma \act X}"'] &
          \Gamma \act TX
          \arrow[d, "\str_{\Gamma, X}"] \\
          &
          T(\Gamma \act X)
        \end{tikzcd}
        \qquad
        \begin{tikzcd}[column sep=large]
          \Gamma \act TTX
          \arrow[r, "\str_{\Gamma, TX}"]
          \arrow[d, "\Gamma \act \mu_X"'] &
          T(\Gamma \act TX)
          \arrow[r, "T\str_{\Gamma, X}"] &
          T(T(\Gamma \act X))
          \arrow[d, "\mu_{\Gamma \act X}"] \\
          \Gamma \act TX
          \arrow[rr, "\str_{\Gamma, X}"'] &&
          T(\Gamma \act X)
        \end{tikzcd}
      \]
    \item liftings of $\act$ to the Kleisli category of $\T$, i.e.\
      actions $\act_{\T}$ of $\V$ on $\Kl \T$ such that the following
      diagram commutes (up to equality, where $K_\T$ is the Kleisli
      inclusion):
      \[
         \begin{tikzcd}
           \V \times \C
           \arrow[r, "\act"]
           \arrow[d, "\V \times K_\T"'] &
           \C
           \arrow[d, "K_\T"] \\
           \V \times \Kl \T
           \arrow[r, "\act_\T"'] &
           \Kl \T
         \end{tikzcd}
      \]
  \end{enumerate}
  If $\S, \T$ are monads on $\C$ equipped with the equivalent data from
  this bijection, then the following conditions on monad morphisms
  $\tau : \S \to \T$ are equivalent: (1) $\tau$ is a strong monad
  morphism $\hat \S \to \hat \T$; (2) $\tau$ is a strong natural
  transformation $\hat S \to \hat T$; (3) $\tau$ makes the diagram on
  the left below commute; (4) $\tau$ makes the diagram on the right
  below commute.
  \[
    \begin{tikzcd}
      \Gamma \act SX
      \arrow[r, "\Gamma \act \tau_X"]
      \arrow[d, "\str_{\Gamma, X}"'] &
      \Gamma \act TX
      \arrow[d, "\str_{\Gamma, X}"] \\
      S(\Gamma \act X)
      \arrow[r, "\tau_{\Gamma \act X}"'] &
      T(\Gamma \act X)
    \end{tikzcd}
    \qquad
    \begin{tikzcd}
      \V \times \Kl \S
      \arrow[r, "\act_\S"]
      \arrow[d, "\V \times \Kl \tau"'] &
      \Kl \S
      \arrow[d, "\Kl \tau"] \\
      \V \times \Kl \T
      \arrow[r, "\act_\T"'] &
      \Kl \T
    \end{tikzcd}
  \]
\end{proposition}
\begin{proof}
  For the bijection between (1) and (2), strong monads induce strong
  functors as above.
  If $\hat T$ is a strong functor with $\under{\hat T} = T$, then
  the unit of $\hat \T$ is $\eta$ and
  the Kleisli extension is given
  by $\extend{f} = \mu_Y \compose \hat T\ctx{\Gamma} f$ for
  $f : \Gamma \act X \to TY$.
  The bijection between (2) and (3) is a special case of
  \cref{prop:strong-functors-strength}, in particular, the two
  diagrams in (3) correspond to $\eta$ and $\mu$ being strong. 
  To go from (3) to (4), define the action $\act_{\T}$ on objects by
  $\Gamma \act_{\T} X = \Gamma \act X$, on $\V$-morphisms by
  $\sigma \act_{\T} X = \sigma \act X$, and on morphisms
  $f \in \Kl \T (X, Y) = \C(X, TY)$ by
  $
    \Gamma \act_{\T} f
      = \str_{\Gamma, Y} \compose (\Gamma \act f)
      \in \Kl \T(\Gamma \act X, \Gamma \act Y)
  $.
  To go from (4) to (3), use $\id_{TX} \in \Kl \T(TX, X)$ to define
  $
    \str_{\Gamma, X} = \Gamma \act_{\T} \id_{TX}
      \in \Kl \T(\Gamma \act TX, \Gamma \act X)
  $.
  The equivalence of the conditions on monad morphisms follows from the
  definition of each bijection.
\end{proof}

The fourth characterization of strong monad is important because of its connection
with the semantics of call-by-value languages in \emph{Freyd
categories}~\cite{power1999closed}.
Indeed, one possible definition of Freyd category explicitly requires
such an action $\act_\T$~\cite{levy2001call}.
The Kleisli inclusion $K_\T$ forms a strong functor with $\act_\T$ as
the action of $\V$ on $\Kl \T$, as does its right adjoint.
If $\tau : \S \to \T$ is a strong
monad morphism, then $\Kl \tau : \Kl \S \to \Kl \T$ also forms a strong
functor.

We again emphasize that strength is additional structure a monad can be
equipped with, not merely a property.
Some monads admit multiple strengths and some admit no strength at all.

\begin{example}\label{example:writer}
  Suppose a monoid $\M = (M, 1, \cdot)$ in $\Set$, and consider the product of right $\M$-actions as an action of the
  Cartesian monoidal category $\Act{\M}$ on itself.
  Equipping $M$ with the discrete action $m \iact m' = m$ 
  makes $\M$ into a monoid in $\Act{\M}$.
  The $\M$-writer monad $\WrM \M$ on $\Act{\M}$ is the functor
  $\Wr M = {-} \times M$ equipped with unit $\eta_X x = (x, 1)$ and
  multiplication $\mu_X ((x, m'), m) = (x, m \cdot m')$.
  If $\M$ is commutative, then $\WrM \M$ forms a strong monad in at least
  two ways.
  As for every writer monad on a monoidal category, the inverse of the
  associator is a strength
  $\str_{\Gamma, X} (\gamma, (x, m)) = ((\gamma, x), m)$; the bijections
  above induce a strong monad in which the Kleisli extension of
  $f : \Gamma \times X \to \Wr M Y$ is given by
  $\extend f (\gamma, (x, m)) = (y, m \cdot m')$ where
  $(y, m') = f(\gamma, x)$.
  Using commutativity, there is also a second strength
  $\str'_{\Gamma, X} (\gamma, (x, m)) = ((\gamma \iact m, x), m)$;
  this induces a strong monad with Kleisli extension
  $\extend f (\gamma, (x, m)) = (y, m \cdot m')$ where
  $(y, m') = f(\gamma \iact m, x)$.

  This example can also be adjusted for the product of sets as an
  action of $(\Act{\M}, \1, \times)$ on $\Set$. In this case,
  commutativity of $\M$ is not needed.
\end{example}

\subsection{Free monads on strong endofunctors}

It is frequently useful to be able to construct the free monad $\T$ on
an endofunctor $F$.
In general, a strength for $F$ will not induce a strength for $\T$; we
give a sufficient condition for this to be the case below.
First we note that, for many applications (even without strength), 
$\T$ being
free (as in \emph{free object}) is not enough.
One often wants the monad $\T$ to be
\emph{algebraically free}~\cite{kelly1980unified}, meaning there is an
isomorphism $\TAlg{\T} \iso \FAlg{F}$ that commutes with the forgetful
functors.
(We write $\TAlg{\T}$ for the Eilenberg-Moore category of the monad
$\T$, and $\FAlg{F}$ for the category of algebras of the functor $F$.)
Algebraic freeness, thus defined, is not a universal property, 
but it still identifies a monad up to a unique isomorphism.
Algebraically free implies free; the converse holds if $\C$ is complete
\cite[Proposition~22.4]{kelly1980unified}.



In general, even the algebraically free monad will not be strong when $F$
has a strength.
To obtain a strength for the monad, we need to refine the notion of free
algebra.
Several versions of the following notion have appeared in the literature
before (for example~\cite{cockett1992strong,pirog2016eilenberg,fiore2017list}).
\begin{definition}\label{def:strongly-free}
  If $F$ is a strong endofunctor on $\C$, an
  $\under F$-algebra $(A, a)$ equipped with a morphism
  $f : X \to A$ is called the \emph{strongly free $F$-algebra}
  on $X \in \C$ if, for all $\Gamma \in \V$, $(B, b) \in \FAlg {\under F}$ and
  $g : \Gamma \act X \to B$,
  there is a unique morphism $h : \Gamma \act A \to B$ such that the
  following diagram commutes.
  \[
    \begin{tikzcd}[column sep=large]
      \Gamma \act X
      \arrow[r, "\Gamma \act f"]
      \arrow[dr, "g"'] 
      &
      \Gamma \act A
      \arrow[d, "h", dashed] 
      &
      \Gamma \act F A
      \arrow[d, "F\ctx{\Gamma} h"]
      \arrow[l, "\Gamma \act a"']\\
      &
      B
      &
      F B
      \arrow[l, "b"']
    \end{tikzcd}
  \]
\end{definition}
If $(A, f, a)$ is the strongly free $F$-algebra on $X$, then it
is also the free $\under F$-algebra on $X$.

\begin{proposition}
  Suppose a strong functor $F$.  If the strongly free $F$-algebra
  $(TX, \eta_X, \sigma_X)$ exists for each object $X$, then $T$ forms
  a strong monad $\T$ in which the unit is $\eta$ and the Kleisli
  extension $\extend g$ of $g : \Gamma \act X \to TY$ is the unique
  morphism $h : \Gamma \act TX \to TY$ such that
\[
    \begin{tikzcd}[column sep=large]
      \Gamma \act X
      \arrow[r, "\Gamma \act \eta_X"]
      \arrow[dr, "g"'] 
      &
      \Gamma \act TX
      \arrow[d, "h", dashed] 
      &
      \Gamma \act F(TX)
      \arrow[d, "F\ctx{\Gamma} h"]
      \arrow[l, "\Gamma \act \sigma_X"']\\
      &
      TY
      &
      F (TY)
      \arrow[l, "\sigma_Y"']
    \end{tikzcd}
  \]
  The monad $\under \T$ is algebraically free on
  $\under F$.
\end{proposition}

It is well-known that, in the presence of right adjoints to
$\Gamma \act {-}$ (in particular, when $\V$ is right closed, in the
case of $\V$ acting on itself), ordinary free algebras 
suffice to construct a strength (see for
example~\cite[Theorem~5]{fiore2008second}); we explain this result in
the context of \emph{powering} in \cref{sec:powering-free}. 
Free algebras are also strongly free when they can be 
constructed as colimits 
that are preserved by $\Gamma \act {-} : \C \to \C$ for
each $\Gamma$.
(See e.g.\ Kelly~\cite{kelly1980unified} for the construction of free
algebras as colimits.)

\subsection{Uniqueness and existence of strengths for monads}
The situation for uniqueness of strengths carries over immediately from
functors (\cref{sec:uniqueness}) to monads.
\begin{proposition}
  Suppose that $\T = (T, \eta, \mu)$ is a monad on $\C$ and that the
  action $\act$ is well-pointed.
  If the functor $T$ forms a strong functor $\hat T$ with $T = \under{\hat T}$
  (necessarily uniquely), then
  $\T$ forms a strong monad $\hat \T$ with $\under{\hat \T} = \T$ (again
  uniquely); moreover, every monad morphism between strong monads is a
  strong monad morphism.
  In particular, if $\act$ is functionally complete, then every monad
  $\T$ on $\C$ forms a strong monad in exactly one way.
\end{proposition}

Existence of strengths for monads is more problematic.
The strengths assigned to the functor $T$ by a weak
functional completeness structure $\Phi$ will not in general make $\T$
into a strong monad.
For example, consider the Cartesian monoidal category $\PSet$ acting on
itself.
This has a single weak functional completeness structure $\Phi$ that assigns to
the identity functor on $\PSet$ the strength
$\str_{\Gamma, X} (\gamma, x) = (*, x)$.
The unit of the identity monad is not a strong
natural transformation with respect to this strength
(its domain is the identity functor with the canonical strength!), so
$\Phi$ does not
make the identity monad into a strong monad.
In fact, if the strength assigned to the identity endofunctor on $\C$ by
a weak functional completeness structure $\Phi$ for an arbitrary action
of $\V$ on $\C$ makes the identity monad into a strong monad, then it
follows that the action is functionally complete.
To see this, note that strong naturality of $\eta$ implies the identity
monad forms a strong monad in only one way: the underlying strong
functor has $\Id\ctx{\Gamma} f = f$.
If $\Phi$ makes the identity monad into a strong monad, we therefore
have
$
  f
  = \Id\ctx{\Gamma} f
  = \Phi_\Gamma{(\lambda \gamma.\,\Id(\pointfun{f}\gamma))}
  = \Phi_\Gamma \pointfun{f}
$,
so $\pointfun{{-}}$ is a bijection, which implies functional
completeness.

\section{Enrichment}\label{sec:enrichment}

So far, we have considered strength only from the perspective of actions
of the monoidal category $\V$.
A well-known result of Kock~\cite{kock1972strong} is that, in a certain situation,
strong functors are the same as enriched functors.
More precisely, Kock shows that if $\V$ is a monoidal category
that is (left) closed in the sense that each ${-} \tensor X : \V \to \V$
has a right adjoint $X \lolly {-} : \V \to \V$, then strengths
$\Gamma \tensor FX \to F(\Gamma \tensor X)$ for an endofunctor
$F : \V \to \V$ are in bijection with
suitable natural transformations $X \lolly Y \to FX \lolly FY$.
The latter make $F$ into an enriched functor $\V \to \V$ (in the sense
of enriched category theory~\cite{kelly1982basic}), where $\V$ enriches
over itself using the closed structure.

It is less well-known that this connection between enrichment and
strength holds more generally.
If $\C$ and $\D$ are any categories that enrich over $\V$, and
suitable adjoints exist, then enriched functors $\C \to \D$ are the same
as strong functors $\C \to \D$.
There are similar bijections for natural transformations, and for
monads.
Strength and enrichment are therefore just two perspectives on the same
structure.
In particular, facts from enriched category theory can be transferred
along these bijections to become facts about strength.

We give the precise connection between strength and enrichment in this
section, again working with an general monoidal category
$\V$. Again, for what we are interested in, we do not need
symmetry.

\begin{definition}
  An \emph{enrichment} of a category $\C$ over a monoidal category $(\V, I, \tensor)$ is a
  functor $\hom : \C^\op \times \C \to \V$ equipped with natural
  transformations
  \[
    \jj_X : I \to X \hom X
    \qquad
    M_{X, Y, Z} : (Y \hom Z) \tensor (X \hom Y) \to X \hom Z
  \]
  such that the functions $\jY_{X,Y} : \C(X, Y) \to \V(I, X \hom Y)$
  given by $\jY_{X,Y} f = (X \hom f) \compose \jj_X$
  are bijections and such that
  the following coherence conditions are satisfied:
\[
\begin{tikzcd}[ampersand replacement=\&]
	{I \tensor (X \hom Y)} \& {(Y \hom Y) \tensor (X \hom Y)} \\
	{X \hom Y} \& {X \hom Y} \\
	{(X \hom Y) \tensor I} \& {(X \hom Y) \tensor (X \hom X)}
	\arrow[Rightarrow, no head, from=2-1, to=2-2]
	\arrow["{M_{X, X, Y}}"', from=3-2, to=2-2]
	\arrow["{(X \hom Y) \tensor \jj_X}"', from=3-1, to=3-2]
	\arrow["{\rho_{X \hom Y}}"', from=2-1, to=3-1]
	\arrow["{\lambda_{X \hom Y}}"', from=1-1, to=2-1]
        \arrow["{\jj_Y \tensor (X \hom Y)}"{yshift=0.5ex}, from=1-1, to=1-2]
	\arrow["{M_{X, Y, Y}}", from=1-2, to=2-2]
\end{tikzcd}
\hspace{-0.4em}
\begin{tikzcd}[column sep=-4.1em, row sep=1.6em, ampersand replacement=\&]
	{((Y \hom Z) \tensor (X \hom Y)) \tensor (W \hom X)} \&\& {(X \hom Z) \tensor (W \hom X)} \\
	{(Y \hom Z) \tensor ((X \hom Y) \tensor (W \hom X))} \&\& {W \hom Z} \\
	\& {(Y \hom Z) \tensor (W \hom Y)}
	\arrow["{\alpha_{Y \hom Z, X \hom Y, W \hom X}}"', from=1-1, to=2-1]
	\arrow["{(Y \hom Z) \tensor M_{W, X, Y}}"', from=2-1, to=3-2]
	\arrow["{M_{W, Y, Z}}"', from=3-2, to=2-3]
        \arrow["{M_{X, Y, Z} \tensor (W \hom X)}"{yshift=0.5ex}, from=1-1, to=1-3]
	\arrow["{M_{W, X, Z}}", from=1-3, to=2-3]
      \end{tikzcd}
    \]
\end{definition}

The objects $X \hom Y$ are the hom-objects of the enrichment; the
natural transformation $j$ gives the identities and $M$ is composition.
Since we do not assume that $\V$ is symmetric, the order of composition
is very important.
The bijection condition in the definition means that the enriched
category has $\C$ as the underlying ordinary category.

\begin{example}
  Recall from \cref{example:copowers} that every monoidal category $\V$
  acts on itself with $\Gamma \act X = \Gamma \tensor X$.
  If each ${{-}} \tensor X$ has a right adjoint $X \lolly {{-}}$,
  then $\lolly$ forms an enrichment of $\V$ over
  itself.
  (This fact is an instance of \cref{prop:action-enrichment} below.)
  This includes for example the category of actions $\Act \M$ of any
  set-theoretic monoid $\M$, with the Cartesian monoidal structure
  (because $\Act \M$ has exponentials).
  Another example is the smash product $\tensor$ of pointed sets, for
  which $X \lolly Y$ is $\PSet(X, Y)$, with distinguished element
  $\lambda x.\,\star$.

  For $\V = \Set$ with the Cartesian monoidal structure, every (locally
  small) $\C$
  has a unique enrichment.
  The object $X \hom Y \in \Set$ is the hom-set $\C(X, Y)$; the
  structural laws $\jj_X : 1 \to \C(X, X)$ and
  $M_{X, Y, Z} : \C(Y, Z) \times \C(X, Y) \to \C(X, Z)$ are the
  identities and composition.
  If $\C$ has copowers over $\Set$ (which form an action of $\Set$ on
  $\C$ as in \cref{example:copowers}), then
  ${-} \copower X \dashv \C(X, {-})$.
\end{example}

\begin{definition}
  If $\C$ and $\D$ are enriched categories, then an \emph{enriched functor} $F : \C \to \D$ consists of an object $FX \in \D$
  for each object $X \in \C$ and a morphism
  $\fmap_{X, Y} : X \hom_\C Y \to FX \hom_\D FY$ for each $X, Y \in \C$,
  such that the following diagrams commute: 
  \[
    \begin{tikzcd}
      I \arrow[r, "\jj_X"]
      \arrow[rd, "\jj_{FX}"']
      &
      X \hom_\C X
      \arrow[d, "\fmap_{X, X}"] \\
      &
      FX \hom_\D FX
    \end{tikzcd}
    \qquad
    \begin{tikzcd}[column sep=huge]
      (Y \hom_\C Z) \tensor (X \hom_\C Y)
      \arrow[r, "M_{X, Y, Z}"]
      \arrow[d, "\fmap_{Y, Z} \tensor \fmap_{X, Y}"'] &
      X \hom_\C Z
      \arrow[d, "\fmap_{X, Z}"] \\
      (FY \hom_\D FZ) \tensor (FX \hom_\D FY)
      \arrow[r, "M_{FX, FY, FZ}"'] &
      FX \hom_\D FZ
    \end{tikzcd}
  \]
  If $F, G : \C \to \D$ are enriched functors, an
  \emph{enriched natural transformation} $\tau : F \natto G$ consists
  of a $\D$-morphism $\tau_x : FX \to GX$ for each $X \in \C$, such that
  the following diagram commutes: 
  \[
    \begin{tikzcd}[column sep=large]
      X \hom_\C Y
      \arrow[r, "\fmap_{X, Y}"]
      \arrow[d, "\fmap_{X, Y}"'] &
      FX \hom_\D FY
      \arrow[d, "FX \hom \tau_Y"] \\
      GX \hom_\D GY
      \arrow[r, "\tau_X \hom GY"'] &
      FX \hom_\D GY
    \end{tikzcd}
  \]
\end{definition}
In the case of $\V = \Set$ with the Cartesian monoidal structure,
enriched functors and natural transformations are just the same as
ordinary functors and natural transformations.  This is a counterpart
to the fact that ordinary functors and natural transformations are
uniquely strong with respect to $\copower$ (if $\D$ has copowers).

The connection between enrichment and strength is the following.
\begin{proposition}\label{prop:action-enrichment}
  Suppose, for each $X \in \C$, a functor ${-} \act X : \V \to \C$ with
  a right adjoint $X \hom {-} : \C \to \V$.
  Also suppose that ${-} \tensor \Gamma$ has a right adjoint
  $\Gamma \lolly {-} : \V \to \V$ for each $\Gamma \in \V$.%
  \footnotemark{}
  Then there is a bijection between (1) the additional data required for
  $\act$ to form an action of $\V$ on $\C$ and (2) the additional data
  required for $\hom$ to form an enrichment of $\C$ over $\V$ such that
  the morphisms $(\Gamma \act X) \hom Y \to \Gamma \lolly (X \hom Y)$,
  obtained from $M$ by transposition, are isomorphisms.
  If both $\C$ and $\D$ are equipped with an action and an enrichment related by
  this bijection, then strong functors $\C \to \D$ are in bijection
  with enriched functors $\C \to \D$.
  Moreover, if $F, G$ are strong, then natural transformations
  $\tau : \under F \natto \under G$ are strong if and only if they
  are enriched.
\end{proposition}
This proposition is not new.
A proof is given for the more general case of enrichment in a bicategory
by Gordon and Power~\cite{gordon1997enrichment}.
Janelidze and Kelly~\cite{janelidze2001note} also give a proof of the
first part of this proposition for enrichment in a monoidal category;
they describe the construction of the enrichment as
``often-rediscovered folklore''.
We sketch the proof here.
\footnotetext{%
  We do not claim that it is \emph{necessary} for $\lolly$ to exist in
  order to connect strength and enrichment, but the statement of this
  proposition is complicated without $\lolly$.
}
\begin{proof}
  It is a standard fact about adjunctions that making
  $\act$ into a bifunctor is equivalent to making
  $\hom$ into a bifunctor.
  By transposition, morphisms $\lambda_X : I \act X \to X$
  are in bijection with morphisms $\jj_X : I \to X \hom X$, and
  $\lambda_X$ is an isomorphism exactly when
  $\jY_{X,Y} : \C(X, Y) \to \V(I, X \hom Y)$ is a
  bijection for all $Y$.
  Families of morphisms
  $
    \alpha_{\Gamma', \Gamma, X}
      : (\Gamma' \tensor \Gamma) \act X
      \to \Gamma' \act (\Gamma \act X)
  $
  natural in $\Gamma, \Gamma'$ are in bijection with
  families of morphisms
  $M_{X,Y,Z} : (Y \hom Z) \tensor (X \hom Y) \to X \hom Z$
  natural in $Y, Z$ by the Yoneda lemma and transposition,
  and $\alpha_{\Gamma, \Gamma', X}$ is invertible for all $\Gamma,
  \Gamma'$ exactly when the morphisms
  $(\Gamma \act X) \hom Y \to \Gamma \lolly (X \hom Y)$
  induced by $M$ are invertible for all $\Gamma, Y$.
  Each of the coherence laws of an action corresponds to one of the laws
  of an enrichment.

  Each strong functor $F : \C \to \D$ comes with functions
  $\C(\Gamma \actC X, Y) \to \D(\Gamma \actD FX, FY)$ natural in
  $\Gamma$.
  By transposition, natural transformations of this type
  are in bijection with $\Gamma$-natural transformations
  $\V(\Gamma, X \hom_\C Y) \to \V(\Gamma, FX \hom_\D FY)$, hence, by the
  Yoneda lemma, with morphisms $X \hom_\C Y \to FX \hom_\D FY$.
  The axioms of enriched and strong functors transfer along this
  bijection, as do strong and enriched naturality.
\end{proof}

\begin{remark}
  There is a more conceptual (and more technical) proof, which we
  outline.
  Wood~\cite{wood1976indical} shows that the 2-category of categories
  enriched over $\V$ embeds fully faithfully into that of what he calls
  \emph{large $\V$-categories}.\footnote{Following Levy~\cite{Lev:locgc}, we prefer to call them \emph{locally} $\V$-\emph{graded categories}.}
  Modulo size issues, these are
  categories enriched over $[\V^\op, \Set]$ with the convolution
  monoidal structure.
  There is a similar embedding of categories equipped with actions of
  $\V$ in large $\V$-categories.
  By characterizing the images of these embeddings, it is possible to
  transfer data between the action perspective and enrichment
  perspective under the assumptions of \cref{prop:action-enrichment}.
  This also works for the powered categories of \cref{sec:powering}.
  Large $\V$-categories then provide a perspective on strength that
  strictly subsumes all of the three perspectives we consider here.
  The \emph{locally indexed categories} used by Levy~\cite{levy2001call}
  and Egger et al.~\cite{egger2009enriching} for strength with respect
  to Cartesian products are similar (but not quite identical) to large
  $\V$-categories; the $[\V^\op, \Set]$ perspective is used by
  Melli\`es~\cite{mellies2012parametric}.
\end{remark}

One of the advantages of considering enrichment is that the concept of
\emph{enriched monad} (corresponding to strong monad) admits a particularly
lightweight definition.
\begin{definition}
   If $\C$ is an enriched category, then an \emph{enriched monad} on $\C$ consists of an object $TX \in \C$
   and morphism $\eta_X : X \to TX$ for each $X \in \C$, and a morphism
   $\bind_{X, Y} : X \hom TY \to TX \hom TY$ for each $X, Y \in \C$,
   such that the following diagrams commute:
   \[
     \begin{tikzcd}[column sep=tiny, row sep=3.5em]
       X {\hom} TY
       \arrow[r, "\phantom{\jj_X}", phantom]
       \arrow[rd, equals]
       \arrow[d, "\bind_{X, Y}"'] &
       \phantom{X \hom X}
       \\
       TX {\hom} TY
       \arrow[r, "\eta_X \hom TY"'{yshift=-0.5ex}] &
       X {\hom} TY
     \end{tikzcd}
     \hspace{-1.4em}
     \begin{tikzcd}[column sep=small, row sep=1.1em]
       I
       \arrow[rdd, "\jj_{TX}"']
       \arrow[r, "\jj_X"] &
       X {\hom} X
       \arrow[d, "X \hom \eta_X"] \\
       &
       X {\hom} TX
       \arrow[d, "\bind_{X, X}"] \\
       &
       TX {\hom} TX
     \end{tikzcd}
     \hspace{1.6em}
     \begin{tikzcd}[row sep=1em,column sep=1.5em]
       (Y {\hom} TZ) \tensor (X {\hom} TY)
       \arrow[r, "\bind_{Y, Z} \tensor (X {\hom} TY)"{inner sep=0,yshift=1ex}]
       \arrow[dd, "\bind_{Y, Z} \tensor \bind_{X, Y}"'] &
       (TY {\hom} TZ) \tensor (X {\hom} TY)
       \arrow[d, "M_{X, TY, TZ}"] \\
       &
       X {\hom} TZ
       \arrow[d, "\bind_{X, Z}"] \\
       (TY {\hom} TZ) \tensor (TX {\hom} TY)
       \arrow[r, "M_{TX, TY, TZ}"'] &
       TX {\hom} TZ
     \end{tikzcd}
   \]
   An \emph{enriched monad morphism} $\tau : \S \to \T$ consists of
   a morphism $\tau_X : SX \to TX$ for each $X \in \C$, such that
   $\tau_X \compose \eta_X = \eta_X$, and such that the following diagram
   commutes:
   \[
     \begin{tikzcd}[row sep=tiny]
       X \hom SY
       \arrow[r, "\bind_{X,Y}"]
       \arrow[dd, "X \hom \tau_Y"'] &
       SX \hom SY
       \arrow[rd, "SX \hom \tau_Y"]
       \\
       &&
       SX \hom TY \\
       X \hom TY
       \arrow[r, "\bind_{X,Y}"'] &
       TX \hom TY
       \arrow[ru, "\tau_X \hom TY"']
     \end{tikzcd}
   \]
\end{definition}

\begin{remark}
  In Haskell (and similar languages), the \verb+Monad+ type class asks
  for a polymorphic function \verb+(>>=) :: m a -> (a -> m b) -> m b+.
  This corresponds to the natural transformation $\bind$ above (with
  arguments reversed).
  Instances of \verb+Monad+ are actually enriched monads (and by the
  following proposition, strong monads), not ordinary monads,
  which is why there is no need to provide a strength in Haskell.
  The same goes for the \verb+Functor+ type class:
  \verb+fmap :: (a -> b) -> (f a -> f b)+ is enriched functoriality, not
  ordinary functoriality.
\end{remark}

\begin{proposition}
  Assume the setting of \cref{prop:action-enrichment},
  with an action of $\V$ on $\C$ that has a corresponding
  enrichment.
  There is a bijection between strong monads on $\C$ and
  enriched monads on $\C$, and this bijection preserves the underlying
  ordinary monads.
  If $\T, \S$ are strong and $\tau : \under \T \to \under \S$ is a
  monad morphism, then $\tau$ is strong if and only if it is
  enriched.
\end{proposition}

Note that the definition of enriched monad involves only 3
equations whereas the definition of strong monad has 4 and the
definition of monad with a
strength has as many as 12 (7 equations of a monad and 5 equations of
a strength for a monad).

\section{Powering}\label{sec:powering}

We now turn to the final perspective on strength that we consider.
Enrichment fits into the picture by considering right adjoints
$X \hom {-}$ to ${-} \act X$.
If instead the functors $\Gamma \act {-} : \C \to \C$ have right
adjoints $\Gamma \power {-} : \C \to \C$, then they form a
\emph{powering} of $\C$ over $\V$ in the following sense; we call
$\Gamma \power X$ the \emph{power} of $\Gamma$ and $X$.\footnotemark{}
\footnotetext{%
  The terminology here comes from the fact that, just as
  $\Gamma \act X$ is a copower (tensor) in the enriched sense when
  $\hom$ exists, $\Gamma \power X$ is a power (cotensor) in the enriched
  sense.}

\begin{definition}
  A \emph{powering} of a category $\C$ over a monoidal category $(\V, I, \tensor)$ is a functor
  $\power : \V^\op \times \C \to \C$ equipped with natural isomorphisms
  \[
    \iR_X : X \to I \power X
    \qquad
    \pR_{\Gamma, \Gamma', X}
      : \Gamma \power (\Gamma' \power X)
      \to (\Gamma' \tensor \Gamma) \power X
  \]
      satisfying the following coherence conditions:
      \[
\begin{tikzcd}[ampersand replacement=\&]
	{\Gamma \power (I \power X)} \& {(I \tensor \Gamma) \power X} \\
	{\Gamma \power X} \& {\Gamma \power X} \\
	{I \power (\Gamma \power X)} \& {(\Gamma \tensor I) \power X}
	\arrow["{\iR_{\Gamma \power X}}"', from=2-1, to=3-1]
	\arrow["{\pR_{I, \Gamma, X}}"', from=3-1, to=3-2]
	\arrow["{\rho_{\Gamma} \power X}"', from=3-2, to=2-2]
	\arrow[Rightarrow, no head, from=2-1, to=2-2]
	\arrow["{\Gamma \power \iR_X}", from=2-1, to=1-1]
	\arrow["{\lambda_\Gamma \power X}"', from=2-2, to=1-2]
	\arrow["{\pR_{\Gamma, I, X}}", from=1-1, to=1-2]
\end{tikzcd}
\hspace{-1.3em}
\begin{tikzcd}[column sep=-4.5em, row sep=1.6em, ampersand replacement=\&]
  \&\& {\Gamma_3 \power ((\Gamma_1 \tensor \Gamma_2) \power X)} \\
  \mathrlap{\hspace{1.6em}\Gamma_3 \power (\Gamma_2 \power (\Gamma_1 \power X))}
  \phantom{\Gamma_3 \power (\Gamma_2 \power (\Gamma_1 \power X))} \&\&\&\&
  \phantom{((\Gamma_1 \tensor \Gamma_2) \tensor \Gamma_3) \power X}
  \mathllap{((\Gamma_1 \tensor \Gamma_2) \tensor \Gamma_3) \power X\hspace{1.6em}} \\
  \& {(\Gamma_2 \tensor \Gamma_3) \power (\Gamma_1 \power X)} \&\& {(\Gamma_1 \tensor (\Gamma_2 \tensor \Gamma_3)) \power X}
  \arrow["{\Gamma_3 \power \pR_{\Gamma_2, \Gamma_1, X}}", from=2-1, to=1-3]
  \arrow["{\pR_{\Gamma_3, \Gamma_1 \tensor \Gamma_2, X}}", from=1-3, to=2-5]
  \arrow["{\alpha_{\Gamma_1, \Gamma_2, \Gamma_3} \power X}"', from=3-4, to=2-5]
  \arrow["{\pR_{\Gamma_2 \tensor \Gamma_3, \Gamma_1, X}}"'{yshift=-1ex}, from=3-2, to=3-4]
  \arrow["{\pR_{\Gamma_3, \Gamma_2, \Gamma_1 \power X}}"', from=2-1, to=3-2]
\end{tikzcd}
\]
\end{definition}

\begin{example}
  If $\V$ is right closed in the sense that each
  $\Gamma \tensor {-} : \V \to \V$ has a right adjoint
  $\Gamma \lollyR {-} : \V \to \V$, then $\lollyR$ gives a powering of $\V$
  over itself (by \cref{prop:action-powering} below).  This right
  adjoint is naturally isomorphic to $\Gamma \lolly {-}$ exactly
  when $\V$ is symmetric. Even when $\V$ is symmetric, the
  definitions of powered functor and powered monad are different from
  the enriched versions (but they are in bijection).

  If a category $\C$ has small products, then it is powered over $\Set$
  by defining $\Gamma \power X = \linebreak \Gamma \pitchfork X = \prod_{\gamma \in \Gamma} X$.
  If $\C$ also has small coproducts, then we have adjunctions
  $\Gamma \copower {-} \dashv \Gamma \pitchfork {-}$.
\end{example}

We define powered notions of functor and natural transformation
analogous to the strong and enriched notions.
\begin{definition}
  If $\C$ and $\D$ are powered categories, then a \emph{powered functor} $F : \C \to \D$ consists of an object
  $FX \in \D$ for each $X \in \C$, and a $\D$-morphism
  $F\ctx{\Gamma} f : FX \to \Gamma \powerD FY$ for each $\C$-morphism
  $f : X \to \Gamma \powerC Y$ such that
  $F\ctx{\Gamma}$ is natural in $\Gamma \in \V$ 
and
\[
\begin{array}{@{}c@{\quad}l}
  F\ctx{I} \iR_X = \iR_{FX} &\textrm{for $X \in \C$}
\\
    F\ctx{\Gamma' \tensor \Gamma}(\pR_{\Gamma, \Gamma', Z}
      \compose (\Gamma \powerC g) \compose f)
    =
    \pR_{\Gamma, \Gamma', FZ} \compose
      (\Gamma \powerD F\ctx{\Gamma'}g) \compose F\ctx{\Gamma} f
  &\textrm{for $f : X\to \Gamma{\powerC}Y$, $g : Y\to \Gamma'{\powerC}Z$}
\end{array}
\]
  If $F, G : \C \to \D$ are powered functors, then a \emph{powered natural
  transformation} $\tau : F \natto G$ consists of a $\D$-morphism
  $\tau_X : FX \to GX$ for each $X \in \C$ such that
  $
    (\Gamma \powerD \tau_Y) \compose F\ctx{\Gamma} f
    =
    G\ctx{\Gamma} f \compose \tau_X
  $
  for $f : X \to \Gamma \powerC Y$.
\end{definition}
If $F : \C \to \D$ is a powered functor, then we obtain an ordinary
functor $\under F : \C \to \D$ by defining $\under F X = FX$ and
$\under F f = \iR^{-1}_{FY} \compose F\ctx{I} (\iR_Y \compose f)$
for $f : X \to Y$.

Equivalently, a powered functor is an ordinary functor $F$ with a
powering, i.e.\ family of morphisms
$\textrm{pow}_{\Gamma,Y} : F (\Gamma \powerC Y) \to \Gamma \powerD FY$
natural in $\Gamma \in \V$ and $Y \in \C$, subject to two equations. We
will not discuss this definition further.

The relationship between strength and powering is as follows.
Similar to the relationship between strength and enrichment
(\cref{prop:action-enrichment}), this proposition enables us to look at
strength from the perspective of powering.
\begin{proposition}\label{prop:action-powering}
  Suppose for each $\Gamma \in \V$ an adjunction
  $\Gamma \act {-} \dashv \Gamma \power {-} : \C \to \C$.
  There is a bijection between the additional data required for
  $\act$ to form an action of $\V$ on $\C$ and the additional data
  required for $\power$ to form a powering of $\C$ over $\V$.
  If both $\C$ and $\D$ are equipped with an action and a powering related by this
  bijection, then there is a bijection between strong functors and
  powered functors $\C \to \D$; this preserves the underlying ordinary
  functors.
  Under this bijection, natural transformations are strong if and only if
  they are powered.
\end{proposition}
We can connect enrichment and powering by combining this proposition with
\cref{prop:action-enrichment}, but also directly by a natural isomorphism
$\V(\Gamma, X \hom Y) \iso \C(X, \Gamma \power Y)$; we omit the precise
statement.

\begin{definition}
  If $\C$ is a powered category, then a \emph{powered monad} $\T = (T, \eta, \extend{({-})})$ consists of an
  object $TX \in \C$ and morphism $\eta_X : X \to TX$ for each
  $X \in \C$, and a morphism $\pextend f : TX \to \Gamma \power TY$ for
  each $f : X \to \Gamma \power TY$, such that $\pextend{(-)}$
  is natural in $\Gamma$ and
\[
\begin{array}{@{}c@{\quad}l}
    \pextend{(\iR_{TX} \compose \eta_X)}
    = \iR_{TX}
    & \textrm{for~} X \in \C
    \\
    \pextend f \compose \eta_X
    = f
    & \textrm{for~} f : X \to \Gamma \power TY
    \\
    \pR_{\Gamma, \Gamma', TZ}
      \compose (\Gamma {\power} \pextend g) \compose \pextend f
    =
    \pextend{(\pR_{\Gamma, \Gamma', TZ}
      \compose (\Gamma {\power} \pextend g) \compose f)}
    & \textrm{for~} f : X \to \Gamma {\power} TY$, $g : Y\to \Gamma' {\power} TZ
\end{array}
\]
  If $\S$ and $\T$ are powered monads, then a \emph{powered monad morphism}
  $\tau : \S \to \T$ consists of a morphism $\tau_X : SX \to TX$ for
  each $X \in \C$ such that $\tau_X \compose \eta_X = \eta_X$ for each
  $X \in \C$ and such that
  $
    (\Gamma \power \tau_Y) \compose \pextend f
      = \pextend{((\Gamma \power \tau_Y) \compose f)} \compose \tau_X
  $
  for each $f : X \to \Gamma \power SY$.
\end{definition}

If $\T = (T, \eta, \pextend{({-})})$ is a powered monad, then $T$ forms
a powered functor $\C \to \C$ by defining
$T\ctx{\Gamma} f = \pextend{((\Gamma \power \eta_Y) \compose f)}$ for
each $f : X \to \Gamma \power Y$.
There is also a monad $\under{\T} = (\under T, \eta, \mu)$, with
multiplication
$\mu_X = \iR_{TX}^{-1} \compose \pextend{\iR_{TX}}$.

As for the action perspective, the powering perspective gives rise to
several equivalent notions of monad, given in the following
proposition.
We emphasize the characterization (3) below in particular.
This characterization is useful when the monad $\T$ is constructed so
that the Eilenberg-Moore category matches some particular category
(for example, the models of an algebraic theory); in which case one way
of making $\T$ into a strong monad is to first obtain a powered monad
using (3), and then obtain a strong monad using
\cref{prop:action-powering}.
\begin{proposition}\label{prop:powered-monads}
  For each monad $\T = (T, \eta, \mu)$ on a powered category $\C$, there is a bijection
  between:
  \begin{enumerate}
    \item powered monads $\hat \T$ such that $\under{\hat \T} = \T$;
    \item powered functors $\hat T$ such that $\under{\hat T} = T$ and
      such that $\eta$ and
      $\mu$ are powered natural
      transformations $\Id \natto \hat T$ and 
$\hat T \fcomp \hat T \natto \hat T$;
    \item liftings of $\power$ to the Eilenberg-Moore category of
      $\T$, i.e.\ powerings $\power_\T$ of $\TAlg{\T}$ over $\V$, such
      that the following diagram commutes (up to equality, where $U_\T$ is the forgetful functor).
      \[
        \begin{tikzcd}
          \V^\op \times \TAlg{\T}
          \arrow[r, "\power_\T"]
          \arrow[d, "\V^\op \times U_\T"'] &
          \TAlg{\T}
          \arrow[d, "U_\T"] \\
          \V^\op \times \C
          \arrow[r, "\power"'] &
          \C
        \end{tikzcd}
      \]
  \end{enumerate}
  If $\S, \T$ are monads equipped with the equivalent data from
  this bijection, then the following conditions on monad morphisms
  $\tau : \S \to \T$ are equivalent: (1) $\tau$ is a powered monad
  morphism $\hat \S \to \hat \T$; (2) $\tau$ is a powered natural
  transformation $\hat S \natto \hat T$; (3) $\tau$ makes the
  diagram below commute.
\[
        \begin{tikzcd}
          \V^\op \times \TAlg{\T}
          \arrow[r, "\power_\T"]
          \arrow[d, "\V^\op \times \TAlg{\tau}"'] &
          \TAlg{\T}
          \arrow[d, "\TAlg{\tau}"] \\
          \V^\op \times \TAlg{\S}
          \arrow[r, "\power_\S"'] &
          \TAlg{\S}
        \end{tikzcd}
\]        
\end{proposition}

\subsection{Free monads on powered endofunctors}
\label{sec:powering-free}

As an application of \cref{prop:powered-monads}, we show that, unlike in the
case of strength with respect to an action, if $\T$ is an algebraically free monad on a powered
functor, then $\T$ is powered in a canonical way.
In light of \cref{prop:action-powering},
this explains why algebraic freeness suffices to construct a (left) 
strength with respect to a monoidal right-closed structure.
\begin{proposition}
  If $F$ is a powered endofunctor on a powered category $\C$ and $\T$ is the
  algebraically free monad on $\under F$, then $\T$ forms a powered
  monad. 

  If $\act$ is an action of $\V$ on $\C$, related to $\power$ as in
  \cref{prop:action-powering},
  and $F$ is a strong endofunctor, then every free
  $\under F$-algebra is strongly free.
\end{proposition}
\begin{proof}
  If $\T$ is algebraically free there is an isomorphism
  $\FAlg{\under F} \iso \TAlg{\T}$ that commutes with the forgetful
  functors.
  By \cref{prop:powered-monads}, to make $\T$ into a powered monad it
  therefore suffices to show that the powering $\power$ lifts to
  $\FAlg{\under F}$.
  To do this, define
  $\power_{\under F} : \V^\op\times \FAlg{\under F} \to \FAlg{\under F}$
  on objects by
  $
    \Gamma \power_{\under F} (A, a)
    =
    (\Gamma \power A, (\Gamma \power a)
      \compose F\ctx{\Gamma} \id_{\Gamma \power A})
  $
  and on morphisms by $\sigma \power_{\under F} f = \sigma \power f$.

  Given an action $\act$ as in \cref{prop:action-powering},
  strong functors and monads are in bijection with powered functors and
  monads, so the algebraically free monad $\T$ then forms a strong
  monad.
  To construct the unique maps $\Gamma \act TX \to A$ of
  \cref{def:strongly-free}, we can therefore use
  the strength $\Gamma \act TX \to T(\Gamma \act X)$ and the fact that
  $T(\Gamma \act X)$ is free on $\Gamma \act X$.
\end{proof}

\section{Conclusion}

We have shown and commented on a number of different equivalent
definitions of strong functor and of strong monad, and explained how and
why they arise. These definitions differ significantly in the amount
of data and the equations they involve, and they serve different
applications. We presented some sufficient conditions for uniqueness
and existence of strengths for all functors, in particular the
condition of weak functional completeness, which is new as far as we
know, and some examples of absence and multiplicity of strengths,
which we crafted to demonstrate that these conditions are not
necessary.

There are some questions we could not settle; for example, we could
neither find a Cartesian category with multiple weak functional
completeness structures nor show that there is none.  We would
like to identify interesting examples of unique existence, absence and
multiplicity of strengths for non-symmetric monoidal categories and
non-self-actions.

A finer analysis of strength could proceed from a generally
non-symmetric non-monoidal closed category $\V$, this being the
minimal structure needed for self-enrichment of $\V$. A further
possible direction of refinement would be to work with skew
monoidal/closed categories and actions,
cf. \cite{szlachanyi2017tensor}. There are no immediate indications of
obstacles, but we would also like to find interesting applications of
this level of generality.

\paragraph{Acknowledgements} Both authors were supported by the
Icelandic Research Fund project grant no. 196323-053, T.U. also by the
Estonian Research Council team grant no. PRG1210.

\bibliographystyle{eptcs}
\bibliography{msfp22}


\appendix

\section{Biactions, bistrong functors, commutative monads}
\label{sec:biactions}

\subsection{Biactions, bistrong functors}

If a monoidal category $\V$ acts on a category $\C$ from both the left
and the right, the two actions can be required to agree with each
other.

\begin{definition}
  A \emph{right action} $\ract$ of a monoidal category $(\V, I, \tensor)$ on a category $\C$ consists of a functor $\ract : \C \times \V \to \C$ and natural
  isomorphisms
  \[
    \rho_X : X \to X \ract I
    \qquad
    \alpha_{X, \Gamma,\Gamma'}
      : (X \ract \Gamma) \ract \Gamma'
      \to X \ract (\Gamma \tensor \Gamma')
  \]
      satisfying the following coherence conditions:
     \[
\begin{tikzcd}[ampersand replacement=\&]
	{(X \ract \Gamma) \ract I} \& {X \ract (\Gamma \tensor I)} \\
	{X \ract \Gamma} \& {X \ract \Gamma} \\
	{(X\ract I) \ract \Gamma} \& {X \ract (I \tensor \Gamma)}
	\arrow["{\alpha_{X, \Gamma, I}}", from=1-1, to=1-2]
	\arrow[Rightarrow, no head, from=2-1, to=2-2]
	\arrow["{\alpha_{X, I, \Gamma}}"', from=3-1, to=3-2]
	\arrow["{\rho_X \ract I}"', from=2-1, to=3-1]
	\arrow["{\rho_{X \ract \Gamma}}", from=2-1, to=1-1]
	\arrow["{X \ract \lambda_\Gamma}"', from=3-2, to=2-2]
	\arrow["{X \ract \rho_\Gamma}"', from=2-2, to=1-2]
\end{tikzcd}
        \hspace{-1.4em}
        \begin{tikzcd}[column sep=-3.2em, row sep=1.6em, ampersand replacement=\&]
          \&\&
          (X \ract \Gamma_3) \ract (\Gamma_2 \tensor \Gamma_1)
          \arrow[rrd, "\alpha_{X, \Gamma_3, \Gamma_2\tensor \Gamma_1}"]
          \\
          \mathrlap{\hspace{1.6em}((X \ract \Gamma_3) \ract \Gamma_2) \ract \Gamma_1}
          \phantom{((X \ract \Gamma_3) \ract \Gamma_2) \ract \Gamma_1}
          \arrow[rru, "\alpha_{X \ract \Gamma_3, \Gamma_2, \Gamma_1}"]
          \arrow[rd, "\alpha_{X, \Gamma_3, \Gamma_2} \ract \Gamma_1"'] \&\&\&\&
          \phantom{X \ract (\Gamma_3 \tensor (\Gamma_2 \tensor \Gamma_1))}
          \mathllap{X \ract (\Gamma_3 \tensor (\Gamma_2 \tensor \Gamma_1))\hspace{1.6em}} \\
          \&
          (X \ract (\Gamma_3 \tensor \Gamma_2)) \ract \Gamma_1
          \arrow[rr, "\alpha_{X, \Gamma_3\tensor \Gamma_2, \Gamma_1}"'{yshift=-0.5ex}]
          \&\&
          X \ract ((\Gamma_3 \tensor \Gamma_2) \tensor \Gamma_1)
          \arrow[ru, "X \ract \alpha_{\Gamma_3, \Gamma_2, \Gamma_1}"']
        \end{tikzcd}
      \]
\end{definition}

\begin{definition}
  A \emph{biaction} of a monoidal category $(\V, I, \tensor)$ on a category $\C$
  consists of a left action
  $\lact : \V \times \C \to \C$, a right action
  $\ract : \C \times \V \to \C$, and a natural isomorphism
  $
    \alpha_{\Gamma, X, \Delta} : (\Gamma \lact X) \ract \Delta
      \to \Gamma \lact (X \ract \Delta)
  $
  such that 
\[
\begin{tikzcd}
(I \act X) \ract \Delta \arrow[r, "\alpha_{I,X,\Delta}"]
   \arrow[dr, "\lambda_X \ract \Delta"']
& I \act (X \ract \Delta) \arrow[d, "\lambda_{X \ract \Delta}"]
\\
& X \ract \Delta
\end{tikzcd}
\quad
\begin{tikzcd}
\Gamma \act X 
   \arrow[d, "\rho_{\Gamma \act X}"']
   \arrow[dr, "\Gamma \act \rho_X"]
\\
(\Gamma \act X) \ract I \arrow[r, "\alpha_{\Gamma,X,I}"']
& \Gamma \act (X \ract I) 
\end{tikzcd}
\]
\[
\begin{tikzcd}
((\Gamma \tensor \Gamma') \act X) \ract \Delta 
   \arrow[rr, "\alpha_{\Gamma \tensor \Gamma',X,\Delta}"]
   \arrow[d, "\alpha_{\Gamma,\Gamma',X} \ract \Delta"']
& & (\Gamma \tensor \Gamma') \act (X \ract \Delta) 
   \arrow[d, "\alpha_{\Gamma,\Gamma',X \ract \Delta}"]
\\
(\Gamma \act (\Gamma' \act X)) \ract \Delta 
         \arrow[r, "\alpha_{\Gamma,\Gamma' \act X,\Delta}"']
& \Gamma \act ((\Gamma' \act X) \ract \Delta)
         \arrow[r, "\Gamma \act \alpha_{\Gamma',X,\Delta}"']
  &  \Gamma \act (\Gamma' \act (X \ract \Delta))
\end{tikzcd}
\]
\[
\begin{tikzcd}
((\Gamma \act X) \ract \Delta) \ract \Delta'
         \arrow[r, "\alpha_{\Gamma,X,\Delta} \ract \Delta'"]
   \arrow[d, "\alpha_{\Gamma \act X, \Delta, \Delta'}"']
& (\Gamma \act (X \ract \Delta)) \ract \Delta'
         \arrow[r, "\alpha_{\Gamma,X \ract \Delta, \Delta'}"]
  & \Gamma \act ((X \ract \Delta) \ract \Delta')
   \arrow[d, "\Gamma \act \alpha_{X, \Delta, \Delta'}"]
\\
(\Gamma \act X) \ract (\Delta \tensor \Delta')
   \arrow[rr, "\alpha_{\Gamma,X,\Delta \tensor \Delta'}"']
& & \Gamma \act (X \ract (\Delta \tensor \Delta')) 
\end{tikzcd}
\]
\end{definition}

An example is $\C = \V$ and $\act = \ract = \otimes$.

In a biaction situation, if a functor has both a left strength and a
right strength, these can be required to cohere as follows.

\begin{definition}
  Suppose a biaction of a monoidal category $\V$ on a category $\D$. A
  \emph{bistrength} for a functor $F : \C \to \D$ is a pair of a left
  strength
  $\lstr_{\Gamma, X} : \Gamma \lactD FX \to F(\Gamma \lactC X)$ and a
  right strength
  $\rstr_{X, \Delta} : FX \ractD \Delta \to F(X \ractC \Delta)$ such
  that
  \[\begin{tikzcd}[column sep=large]
	{(\Gamma \lactD FX) \ractD \Delta} & {F(\Gamma \lactC X) \ractD \Delta} & {F((\Gamma \lactC X) \ractC \Delta)} \\
	{\Gamma \lactD (FX \ractD \Delta)} & {\Gamma \lactD F(X \ractC \Delta)} & {F (\Gamma \lactC (X \ractC \Delta))}
	\arrow["{\Gamma \lactD \rstr_{X,\Delta}}"', from=2-1, to=2-2]
	\arrow["\lstr_{\Gamma,X \ractC \Delta}"', from=2-2, to=2-3]
	\arrow["\alpha_{\Gamma,FX,\Delta}"', from=1-1, to=2-1]
	\arrow["{\lstr_{\Gamma,X} \ractD \Delta}", from=1-1, to=1-2]
	\arrow["\rstr_{\Gamma \actC X, \Delta}", from=1-2, to=1-3]
	\arrow["F\alpha_{\Gamma,X,\Delta}", from=1-3, to=2-3]
\end{tikzcd}
\]
A natural transformation between two bistrong functors is bistrong 
if it is both left strong and right strong. 
\end{definition}

Consider the case $\C = \D = \V$ and $\act = \ract = \otimes$.  If
$\V$ is symmetric, with braiding $c_{X,Y} : X \tensor Y \to Y \tensor X$, then any left strength $\lstr$ of a functor $F$
induces a right strength $\rstr$ via
$\rstr_{X,\Delta} = F c_{\Delta,X} \compose \lstr_{\Delta,X} \compose
c_{FX,\Delta}$.
The two strengths together form a bistrength.
But the right strength does not have to be related to the left
strength like this, not even when $\V$ is a Cartesian category. For
example, take $\V$ to be the category of pointed sets with its
Cartesian structure. The identity functor is bistrong with
$\lstr_{\Gamma,X} (\gamma, x) = (\gamma, x)$ and
$\rstr_{X,\Delta} (x, \delta) = (x, \star)$.

\subsection{Commutative monads}

Kock~\cite{kock1970monads} studied what he named commutative monads for the case of a
symmetric monoidal category. His commutative monads were left-strong
monads subject to an additional equational condition.

Symmetry is in fact not needed. The concept of commutative monad makes
sense for a general monoidal category $\V$; Kock's condition can be
formulated for any bistrength for the tensor as a biaction (where the
right strength need not in general be defined in terms of the left
strength like we did above).

\begin{definition}
  Suppose a monoidal category $(\V, I, \tensor)$. A \emph{commutative
    monad} is a monad $\T = (T, \eta, \mu)$ with a bistrength
  $(\lstr, \rstr)$ of $T$ (wrt.\ $\tensor$ as a biaction of $\V$ on
  itself) such that $\eta$, $\mu$ are bistrong and moreover the
  following diagram commutes:
\[
\begin{tikzcd}
T X \tensor T Y
   \arrow[r, "\lstr_{TX,Y}"]
   \arrow[d, "\rstr_{X, TY}"']
& T (TX \tensor Y)  \arrow[r, "T \rstr_{X,Y}"]
&  T (T (X \tensor Y))
   \arrow[d, "\mu_{X \tensor Y}"]
\\
T (X \tensor T Y)
         \arrow[r, "T \lstr_{X,Y}"']
& T (T (X \tensor Y))
         \arrow[r, "\mu_{X \tensor Y}"']
  & T (X \tensor Y)
\end{tikzcd}
\]

\end{definition}
Commutative monads in this sense are exactly the same as lax
monoidal monads.
Even when $\V$ is symmetric, the bistrength of a commutative monad does
not need to be defined by symmetry, so this notion of commutative monad
(i.e.\ lax monoidal monad) is strictly more general than
Kock's.
For example, consider the writer monad $\WrM \M$ on $\Act \M$ from
\cref{example:writer}, where $\M$ is any commutative monoid.
From the two strengths given there, we can make a bistrength
\[
  \lstr_{\Gamma, X} (\gamma, (x, m)) = ((\gamma, x), m)
  \qquad
  \rstr_{X, \Delta} ((x, m), \delta) = ((x, \delta * m), m)
\]
and $\WrM \M$ equipped with this bistrength is a commutative monad.
Kock's commutative monads are the same as symmetric lax monoidal monads.

\end{document}